\documentclass{article}
\usepackage{amssymb}

\usepackage{graphicx}
\usepackage{amsmath}

\newtheorem{theorem}{Theorem}

\newtheorem{definition}[theorem]{Definition}

\newtheorem{lemma}[theorem]{Lemma}

\newtheorem{proposition}[theorem]{Proposition}

\newenvironment{proof}[1][Proof]{\textbf{#1.} }{\ \rule{0.5em}{0.5em}}
\input{tcilatex}

\begin{document}

\author{M. Mond \thanks{%
The Pearlstone Center for Aeronautical Engineering Studies, Department of
Mechanical Engineering, Ben-Gurion University of the Negev, Beer-Sheva,
Israel. E-mail: mond@bgu.ac.il} \and V. S. Borisov \thanks{%
E-mail: viatslav@bgu.ac.il}}
\title{On stability of difference schemes. Central schemes for hyperbolic
conservation laws with source terms}
\maketitle

\begin{abstract}
The stability of difference schemes for, in general, hyperbolic systems of
conservation laws with source terms are studied. The basic approach is to
investigate the stability of a non-linear scheme in terms of its
corresponding scheme in variations. Such an approach leads to application of
the stability theory for linear equation systems to establish stability of
the corresponding non-linear scheme. It is established the notion that a
non-linear scheme is stable if and only if the corresponding scheme in
variations is stable.

A new modification of the central Lax-Friedrichs (LxF) scheme is developed
to be of the second order accuracy. A monotone piecewise cubic interpolation
is used in the central schemes to give an accurate approximation for the
model in question. The stability of the modified scheme are investigated.
Some versions of the modified scheme are tested on several conservation
laws, and the scheme is found to be accurate and robust.

As applied to hyperbolic conservation laws with, in general, stiff source
terms, it is constructed a second order nonstaggered central scheme based on
operator-splitting techniques.
\end{abstract}


\section{Introduction\label{Introduction}}

We are mainly concerned with the stability of difference schemes for
hyperbolic systems of conservation laws with source terms. Such systems are
used to describe many physical problems of great practical importance in
magneto-hydrodynamics, kinetic theory of rarefied gases, linear and
nonlinear waves, viscoelasticity, multi-phase flows and phase transitions,
shallow waters, etc. (see, e.g., \cite{Bereux and Sainsaulieu 1997}, \cite
{Caflisch at el. 1997}, \cite{Godlewski and Raviart 1996}, \cite{Jin Shi
1995}, \cite{Kurganov and Tadmor 2000}, \cite{LeVeque 2002}, \cite{Monthe
2003}, \cite{Naldi and Pareschi 2000}, \cite{Pareschi Lorenzo 2001}, \cite
{Pareschi and Russo 2005}). We will consider a system of hyperbolic
conservation laws written as follows (e.g., \cite{Godlewski and Raviart 1996}%
, \cite{LeVeque 2002}) 
\begin{equation}
\frac{\partial \mathbf{u}}{\partial t}+\sum_{j=1}^{N}\frac{\partial }{%
\partial x_{j}}\mathbf{f}_{j}\left( \mathbf{u}\right) =\frac{1}{\tau }%
\mathbf{q}\left( \mathbf{u}\right) ,\ 0<t\leq T_{\max },\ \left. \mathbf{u}%
\left( \mathbf{x},t\right) \right| _{t=0}=\mathbf{u}^{0}\left( \mathbf{x}%
\right) ,  \label{INA10}
\end{equation}
where $\mathbf{x\equiv }\left\{ x_{1},x_{2},\ldots ,x_{N}\right\} ^{T}\in 
\mathbb{R}^{N}$, $\mathbf{u}=\left\{ u_{1},u_{2},\ldots ,u_{M}\right\} ^{T}$
is a vector-valued function from $\mathbb{R}^{N}$ $\times $ $[0,+\infty )$
into an open subset $\Omega _{\mathbf{u}}\subset \mathbb{R}^{M}$, $\mathbf{f}%
_{j}\left( \mathbf{u}\right) =\left\{ f_{1j}\left( \mathbf{u}\right) \right.
,$ $f_{2j}\left( \mathbf{u}\right) ,$ $\ldots ,$ $\left. f_{Mj}\left( 
\mathbf{u}\right) \right\} ^{T}$ is a smooth function (flux-function) from $%
\Omega _{\mathbf{u}}$ into $\mathbb{R}^{M}$, $\mathbf{q}\left( \mathbf{u}%
\right) =\left\{ q_{1}\left( \mathbf{u}\right) ,q_{2}\left( \mathbf{u}%
\right) ,\ldots ,q_{M}\left( \mathbf{u}\right) \right\} ^{T}$ denotes the
source term, $\tau >0$ denotes the stiffness parameter, $\mathbf{u}%
^{0}\left( \mathbf{x}\right) $\ is of compact support. We will assume that $%
\tau =const$ without loss of generality. In what follows $\left\| \mathbf{M}%
\right\| _{p}$ denotes the matrix norm of a matrix $\mathbf{M}$ induced by
the vector norm $\left\| \mathbf{v}\right\| _{p}$ $=\left( \sum_{i}\left|
v_{i}\right| ^{p}\right) ^{1/p}$, and $\left\| \mathbf{M}\right\| $ denotes
the matrix norm induced by a prescribed vector norm. $\mathbb{R}$ denotes
the field of real numbers.

For studying stability and monotonicity of non-linear schemes, the well
known notion of total variation diminishing (TVD, see, e.g., \cite{Godlewski
and Raviart 1996}, \cite{LeVeque 2002}) turns out to be an useful tool.
Actually, the following property

\begin{equation}
\left\| \mathcal{N}\left( \mathbf{v}+\delta \mathbf{v}\right) -\mathcal{N}%
\left( \mathbf{v}\right) \right\| \leq \left( 1+\alpha \Delta t\right)
\left\| \delta \mathbf{v}\right\|  \label{IN50}
\end{equation}
is sufficient for stability of a two-step method \cite{LeVeque 2002},
however it is, in general, difficult to obtain. Here $\Delta t$ denotes the
time increment, $\alpha $ is a constant independent of $\Delta t$ as $\Delta
t\rightarrow 0$, $\mathbf{v}$ and $\delta \mathbf{v}$ are any two grid
functions ($\delta \mathbf{v}$ will often be referred to as the variation of
the grid function $\mathbf{v}$), $\mathcal{N}$ denotes the scheme operator.
At the same time, the stability of linearized version of the non-linear
scheme is generally not sufficient to prove convergence \cite{Harten 1984}, 
\cite{LeVeque 2002}. Instead, the TV-stability adopted in \cite{Harten 1984}
(see also \cite[s. 8.3.5]{LeVeque 2002}) makes it possible to prove
convergence (to say, TV-convergence) of non-linear scalar schemes with ease.
However, the TVD property is a purely scalar notion that cannot, in general,
be extended for non-linear systems of equations, as the true solution itself
is usually not TVD \cite{Godlewski and Raviart 1996}, \cite{LeVeque 2002}.
Moreover, one can see in \cite[pp. 1578-1581]{Borisov and Sorek 2004} that a
TVD scheme can be non-convergent in, at least, $L_{\infty }$, in spite that
the scheme is TV-stable. Such a phenomenon is, in all likelihood, caused by
the fact that TV is not a norm, but a semi-norm. 

Nowadays, there exists a few methods for stability analysis of some classes
of nonlinear difference schemes approximating systems of PDEs (see, e.g., 
\cite{Ganzha and Vorozhtsov 1996b}, \cite{Gil' 2007}, \cite{LeVeque 2002}, 
\cite{Morton 1996}, \cite{Naterer and Camberos 2008}, \cite{Samarskii 2001}
and references therein). It is noted in \cite{Gil' 2007} that the problem of
stability analysis is still one of the most burning problems, because of the
absence of its complete solution. In particular, as noted in \cite{Ganzha
and Vorozhtsov 1996b} in this connection, the vast majority of difference
schemes, currently in use, have still not been analyzed. LeVeque \cite
{LeVeque 2002} noted as well that, in general, no numerical method for
non-linear systems of equations has been proven to be stable. There is not
even a proof that the first-order Godunov method converges on general
systems of non-linear conservation laws \cite[p. 340]{LeVeque 2002}. Thus, a
different approach to testing scheme stability must be adopted to prove
convergence of non-linear schemes for systems of PDEs. The notion of scheme
in variations (or variational scheme \cite{Borisov and Sorek 2004}, \cite
{Borisov V.S. 2003}) has, in all likelihood, much potential to be an
effective tool for studying stability of nonlinear schemes. Such an approach
goes back to the one suggested by Lyapunov (1892), namely, to investigate
stability by the first approximation. This idea has long been exploited for
investigation of the stability of motion \cite{Gil' 1998}. An approach to
investigate non-linear difference schemes for monotonicity in terms of
corresponding variational schemes was suggested in \cite{Borisov and Sorek
2004}, \cite{Borisov V.S. 2003}. The advantage of such an approach is that
the variational scheme will always be linear and, hence, enables the
investigation of the monotonicity for nonlinear operators using linear
patterns. It is proven for the case of explicit schemes that the
monotonicity of a variational scheme will guarantee that its original scheme
will also be monotone \cite{Borisov and Sorek 2004}. We establish the notion
that the stability of a scheme in variations is necessary and sufficient for
the stability of its original scheme (see Section \ref{Stability of
difference schemes}, Theorem \ref{iffStab}).

An extensive literature is devoted to central schemes, since these schemes
are attractive for various reasons: no Riemann solvers, characteristic
decompositions, complicated flux splittings, etc., must be involved in
construction of a central scheme (see, e.g., \cite{Balaguer and Conde 2005}, 
\cite{Kurganov and Tadmor 2000}, \cite{Kurganov and Levy 2000}, \cite
{LeVeque 2002}, \cite{Pareschi Lorenzo 2001}, \cite{Pareschi et al. 2005}
and references therein), and hence such schemes can be implemented as a
black-box solvers for general systems of conservation laws \cite{Kurganov
and Tadmor 2000}. Let us, however, note that the numerical domain of
dependence \cite[p. 69]{LeVeque 2002} for a central scheme approximating,
e.g., a scalar transport equation coincides with the numerical domain of
dependence for a standard explicit scheme approximating diffusion equations 
\cite[p. 67]{LeVeque 2002}. Such a property is inherent to central schemes
in contrast to, e.g., the first-order upwind schemes \cite[p. 73]{LeVeque
2002}. Hence, central schemes do not satisfy the long known principle (e.g., 
\cite[p. 304]{Anderson et al. 1984.}) that derivatives must be correctly
treated using type-dependent differences, and hence there is a risk for
every central scheme to exhibit spurious solutions. The results of
simulations in \cite{Nessyahu and Tadmor 1990} can be seen as an
illustration of the last assertion. Notice, all versions of the, so called,
Nessyahu-Tadmor (NT) central scheme, in spite of sufficiently small CFL
(Courant-Friedrichs-Lewy \cite{LeVeque 2002}) number ($Cr=0.475$), exhibit
spurious oscillations in contrast to the second-order upwind scheme ($%
Cr=0.95 $). The first order, $O(\Delta t+\Delta x)$, LxF scheme exhibits the
excessive numerical viscosity. Thus, the central scheme should be chosen
with great care to reflect the true solution and to avoid significant but
spurious peculiarities in numerical solutions.

Let us note that LxF scheme -- the forerunner for central schemes \cite
{Balaguer and Conde 2005}, \cite{Kurganov and Tadmor 2000} -- does not
produce spurious oscillations. While, from the pioneering works of Nessyahu
and Tadmor \cite{Nessyahu and Tadmor 1990} and on, the higher order versions
of LxF scheme can produce spurious oscillations. The reason has to do with a
negative numerical viscosity introduced to obtain a higher order accurate
scheme (for more details, see Section \ref{COSN}). Let us note that there is
a possibility to increase the scheme's order of accuracy, up to $O((\Delta
t)^{2}+\left( \Delta x\right) ^{2})$, by introducing an additional
non-negative numerical viscosity into the scheme. Such an approach is
similar to the vanishing viscosity method \cite{Godlewski and Raviart 1996}, 
\cite{LeVeque 2002}, and hence possesses its advantages, yet it appears to
be free of the disadvantages of this method, since the additional viscosity
term is not artificial. With this approach, the second order scheme is
developed in Section \ref{COSN}, where sufficient conditions for stability
of the scheme are found. The scheme is tested on several conservation laws
in Section \ref{Exemplification and discussion}.

A stable numerical scheme may yield spurious results when applied to a stiff
hyperbolic system with relaxation (see, e.g., \cite{Ahmad and Berzins 2001}, 
\cite{Aves Mark A. et al. 2000}, \cite{Bereux and Sainsaulieu 1997}, \cite
{Caflisch at el. 1997}, \cite{Du Tao et al. 2003}, \cite{Jin Shi 1995}, \cite
{Pember 1993}, \cite{Pember 1993a}). Specifically, spurious numerical
solution phenomena may occur when underresolved numerical schemes (i.e.,
insufficient spatial and temporal resolution) are used (e.g., \cite{Ahmad
and Berzins 2001}, \cite{Jin Shi 1995}, \cite{JIN AND LEVERMORE 1996}, \cite
{Naldi and Pareschi 2000}). However, during a computation, the stiffness
parameter may be very small, and, hence, to resolve the small stiffness
parameter, we need a huge number of time and spatial increments, making the
computation impractical. Hence, we are interested to solve the system, (\ref
{INA10}), with underresolved numerical schemes. It is significant that for
relaxation systems a numerical scheme must possess a discrete analogy to the
continuous asymptotic limit, because any scheme violating the correct
asymptotic limit leads to spurious or poor solutions (see, e.g., \cite
{Caflisch at el. 1997}, \cite{Jin Shi 1995}, \cite{Jin Shi et al. 2000}, 
\cite{Naldi and Pareschi 2000}, \cite{Pareschi Lorenzo 2001}). Most methods
for solving such systems can be described as operator splitting ones, \cite
{Du Tao et al. 2003}, or methods of fractional steps, \cite{Bereux and
Sainsaulieu 1997}. After operator splitting, one solves the advection
homogeneous system, and then the ordinary differential equations associated
with the source terms. As reported in \cite{Gosse L. 2000}, this approach is
well suited for the stiff systems. We are mainly concerned with such an
approach in Section \ref{OSS}.

\section{Stability of difference schemes\label%
{Stability of difference schemes}}

Let us consider the following non-linear explicit scheme arising, e.g., in
numerical analysis of nonlinear PDE systems: 
\begin{equation}
\mathbf{v}_{i}^{n+1}=\mathbf{H}_{i}^{n}(\mathbf{v}_{1}^{n},\mathbf{v}%
_{2}^{n},\ldots ,\mathbf{v}_{I}^{n}),\ \mathbf{H}_{i}^{n}:\Omega
_{n}\subseteq \mathbb{R}^{N}\rightarrow \mathbb{R}^{N_{0}},\ i\in \omega
_{1},\ n,n+1\in \omega _{2},  \label{NSCS10}
\end{equation}
where  $\mathbf{v}_{i}^{n}\in \mathbb{R}^{N_{0}}$ denotes a vector-valued
grid function, $N=N_{0}I$, $i\in \omega _{1}$ denotes a node of the grid $%
\omega _{1}\equiv \left\{ 1,2,\ldots ,I\right\} $, $n\in \omega _{2}$
denotes a node (time level) of the grid $\omega _{2}\equiv \left\{
0,1,\ldots ,M\right\} $, $\mathbf{H}_{i}^{n}$$=$ $\left\{
H_{i,1}^{n},\right. H_{i,2}^{n},$ $\ldots ,$ $\left. H_{i,N_{0}}^{n}\right\}
^{T}$ is a vector-valued function with the domain and range belonging to $%
\mathbb{R}^{N}$ and $\mathbb{R}^{N_{0}}$, respectively. Notice, $\mathbf{H}%
_{i}^{n}$ depends also on scheme parameters (e.g., space and time
increments), however, this dependence is usually not included in the
notation. We will assume that $n$\ in (\ref{NSCS10}) denotes the time level, 
$t_{n}$ $\left( =n\Delta t\right) $. Thus, the time increment will be
represented by $\Delta t=t_{\max }/M=const$, where $t_{\max }$\ denotes some
finite time over which we wish to compute. If we introduce the additional
notation 
\begin{equation}
\mathbf{v}^{n}\mathbf{=}\left\{ \left( \mathbf{v}_{1}^{n}\right) ^{T},\left( 
\mathbf{v}_{2}^{n}\right) ^{T},\ldots ,\left( \mathbf{v}_{I}^{n}\right)
^{T}\right\} ^{T},\ \mathbf{H}^{n}\mathbf{=}\left\{ \left( \mathbf{H}%
_{1}^{n}\right) ^{T},\left( \mathbf{H}_{2}^{n}\right) ^{T},\ldots ,\left( 
\mathbf{H}_{I}^{n}\right) ^{T}\right\} ^{T},  \label{NSCS05}
\end{equation}
then the scheme (\ref{NSCS10}) can be written in the form 
\begin{equation}
\mathbf{v}^{n+1}=\mathbf{H}^{n}(\mathbf{v}^{n}),\quad \mathbf{H}^{n}:\Omega
_{n}\subseteq \mathbb{R}^{N}\rightarrow \mathbb{R}^{N},\quad n,n+1\in \omega
_{2}\equiv \left\{ 0,1,\ldots ,M\right\} .  \label{NSCS20}
\end{equation}

As usual (e.g., \cite[p. 62]{Ortega and Rheinboldt 1970}), for mappings $%
\mathbf{f}:\Omega _{f}\subseteq \mathbb{R}^{N}\rightarrow \mathbb{R}^{N}$
and $\mathbf{g}:\Omega _{g}\subseteq \mathbb{R}^{N}\rightarrow \mathbb{R}%
^{N} $, the composite mapping $\mathbf{h}=\mathbf{g}\circ \mathbf{f}$ is
defined by $\mathbf{h}\left( \mathbf{v}\right) =\mathbf{g}\left( \mathbf{f}%
\left( \mathbf{v}\right) \right) $ for all $\mathbf{v\in }\Omega
_{h}=\left\{ \mathbf{v\in }\Omega _{f}\mid \mathbf{f}\left( \mathbf{v}%
\right) \in \Omega _{g}\right\} $. Using the composite mapping approach, we
rewrite Scheme (\ref{NSCS20}) to read 
\begin{equation}
\mathbf{y}=\mathbf{F}\left( \mathbf{x}\right) ,\quad \mathbf{F}:\Omega
_{F}\subseteq \mathbb{R}^{N}\rightarrow \mathbb{R}^{N},  \label{NSCS30}
\end{equation}
where the following notation is used: $\mathbf{x=v}^{0}$, $\mathbf{y}=%
\mathbf{v}^{M}$, $\mathbf{F=H}^{M-1}\circ \mathbf{H}^{M-2}\circ \ldots \circ 
\mathbf{H}^{0}$, $\Omega _{F}=\left\{ \mathbf{v}^{0}\mathbf{\in }\Omega
_{0}\mid \right. $ $\mathbf{v}^{1}=\mathbf{H}^{0}\left( \mathbf{v}%
^{0}\right) \in \Omega _{1}\mid $ $\ldots $ $\mid \mathbf{v}^{M-1}=$ $\left. 
\mathbf{H}^{M-2}\left( \mathbf{v}^{M-2}\right) \in \Omega _{M-1}\right\} $.
Let the scheme parameters (including time increments) be represented by a
vector $\mathbf{s}$ belonging to some normed space with the norm $\left| 
\mathbf{s}\right| $.

Since differentiability of $\mathbf{H}^{n}$ as well as $\mathbf{F}$ will be
a key element in the following, let us note that the composite mapping $%
\mathbf{F}$\ will be Fr\'{e}chet-differentiable \cite[item 3.1.5]{Ortega and
Rheinboldt 1970} if all of the maps, $\mathbf{H}^{n}$, are
Fr\'{e}chet-differentiable \cite[item 3.1.7]{Ortega and Rheinboldt 1970}.\
However, if all of the maps are Fr\'{e}chet-differentiable, but one that
Gateaux-differentiable \cite[item 3.1.1]{Ortega and Rheinboldt 1970}, then
the composite mapping $\mathbf{F}$\ has a Gateaux-derivative \cite[item
3.1.7]{Ortega and Rheinboldt 1970}. Notice, if there exist at least two maps
having Gateaux-derivatives, then $\mathbf{F}$ need not be differentiable 
\cite[E 3.l-7]{Ortega and Rheinboldt 1970}.

Scheme (\ref{NSCS30}) is said to be stable (see, e.g., \cite{Ganzha and
Vorozhtsov 1996b}, \cite{Godlewski and Raviart 1996}, \cite{LeVeque 2002}, 
\cite{Richtmyer and Morton 1967}, \cite{Samarskii 2001}, \cite{Samarskiy and
Gulin 1973}) if there exist positive $s_{0}$, $C=const$ such that for all $%
\mathbf{x,}$ $\mathbf{x}_{\ast }\in \Omega _{F}$ the following inequality is
valid 
\begin{equation}
\left\| \mathbf{F}\left( \mathbf{x}_{\ast }\right) -\mathbf{F}\left( \mathbf{%
x}\right) \right\| \leq C\left\| \mathbf{x}_{\ast }-\mathbf{x}\right\|
,\quad \forall \ \mathbf{s}:\ \left| \mathbf{s}\right| \leq s_{0}.
\label{NSCS40}
\end{equation}
Thus, Scheme (\ref{NSCS30}) will be stable \emph{iff} (if and only if) the
function $\mathbf{F}$ will be Lipschitz for a constant $C$.

To be more specific, let us consider the ``slit plane'' \cite{Heinonen Juha
2005} in polar coordinates $\left( r,\theta \right) $ 
\begin{equation}
\Omega _{F}=\left\{ \left( r,\theta \right) \mid 0<r<\infty ,\ -\pi <\theta
<\pi \right\} \subset \mathbb{R}^{2},  \label{NSCS42}
\end{equation}
and the function $\mathbf{F}=\left\{ F_{1}\left( r,\theta \right)
,F_{2}\left( r,\theta \right) \right\} ^{T}$ such that \cite{Heinonen Juha
2005} 
\begin{equation}
F_{1}=r,\ F_{2}=\theta \diagup 2.  \label{NSCS44}
\end{equation}
If we take $\left( r,\theta \right) _{\ast }=\left( r_{0},-\pi +\varepsilon
\right) $, $\left( r,\theta \right) =\left( r_{0},\pi -\varepsilon \right) $%
, and $r_{0}=const$, then obviously the mapping (\ref{NSCS44}) is not
Lipschitz, since $C$ in (\ref{NSCS40}) tends to infinity\ as $\varepsilon
\rightarrow 0$. Therefore, we have to conclude, in view of the above
definition, that Scheme (\ref{NSCS44}) is not stable, even though the
function $\mathbf{F}$, (\ref{NSCS44}), is locally Lipschitz for $C=1$, and,
further, $\mathbf{F}$ is the non-stretching mapping of the ``slit plane'' (%
\ref{NSCS42}) into the right semi-plane. Hence, the preceding definition of
stability needs minor changes.

A set $\Omega \subseteq \mathbb{R}^{N}$ is said to be path-connected if
every two points $\mathbf{x}$, $\mathbf{x}_{\ast }$ $\in \Omega $ can be
joined by a continuous curve ($\gamma :\left[ 0,1\right] \subset \mathbb{R}%
\rightarrow \Omega $, \cite{Heinonen Juha 2005}, \cite[p. 113]{Kolmogorov
and Fomin 1970}) of finite length, $L\left( \gamma \right) $. The intrinsic
metric \cite{Heinonen Juha 2005} $\Lambda _{\Omega }$ in a path-connected
set $\Omega $\ is defined as 
\begin{equation}
\Lambda _{\Omega }\left( \mathbf{x},\mathbf{x}_{\ast }\right) =\underset{%
\gamma \subset \Omega }{\inf }L\left( \gamma \right) ,\quad \gamma :\ 
\mathbf{x=}\gamma \left( 0\right) ,\ \mathbf{x}_{\ast }=\gamma \left(
1\right) ,\ L\left( \gamma \right) <\infty .  \label{NSCS46}
\end{equation}
An open ball (of radius $r$) about $\mathbf{x\in }\mathbb{R}^{N}$ is denoted
by $B\left( \mathbf{x},r\right) $ (or just $B_{\mathbf{x}}$).

\begin{definition}
\label{DefStability}Let $\Omega _{F}$ in (\ref{NSCS30}) be path-connected.
Scheme (\ref{NSCS30}) is said to be stable if there exist positive $s_{0}$, $%
C=const$ such that the following inequality holds 
\begin{equation}
\left\| \mathbf{F}\left( \mathbf{x}_{\ast }\right) -\mathbf{F}\left( \mathbf{%
x}\right) \right\| \leq C\Lambda _{\Omega _{F}}\left( \mathbf{x},\mathbf{x}%
_{\ast }\right) ,\quad \forall \ \mathbf{x,x}_{\ast }\in \Omega _{F},\quad
\forall \ \mathbf{s}:\ \left| \mathbf{s}\right| \leq s_{0}.  \label{NDS20}
\end{equation}
\end{definition}

Notice, Scheme (\ref{NSCS44}) is stable, since Inequality (\ref{NDS20})
holds for $C=1$.

\begin{lemma}
\label{LipschitzStability}Let the path-connected $\Omega _{F}$ of (\ref
{NSCS30}) be open in $\mathbb{R}^{N}$. Scheme (\ref{NSCS30}) will be stable
in terms of Definition \ref{DefStability} \emph{iff}\ $\mathbf{F}$\ in (\ref
{NSCS30}) will be locally Lipschitz for a common constant $C$, for all
scheme parameters $\mathbf{s}$ such\ that $\left| \mathbf{s}\right| \leq
s_{0}$.
\end{lemma}

\begin{proof}
Suppose Scheme (\ref{NSCS30}) is stable, i.e. (\ref{NDS20}) is valid. Choose
any point $\mathbf{x\in }\Omega _{F}$. Since $\Omega _{F}$ is open, there
exists a radius $r$ such that $B\left( \mathbf{x},r\right) \subset \Omega
_{F}$. Choose any point $\mathbf{x}_{\ast }\in B\left( \mathbf{x},r\right) $%
, and let $\gamma _{\ast }$ be the straight line segment joining the points $%
\mathbf{x}$, $\mathbf{x}_{\ast }$ $\in $ $B\left( \mathbf{x},r\right) $. In
view of (\ref{NDS20}), $\mathbf{F}$\ in (\ref{NSCS30}) will be locally
Lipschitz for a common constant $C$, for all $\mathbf{s}:\ \left| \mathbf{s}%
\right| \leq s_{0}$, since $\Lambda _{\Omega _{F}}\left( \mathbf{x},\mathbf{x%
}_{\ast }\right) =L\left( \gamma _{\ast }\right) =\left\| \mathbf{x}_{\ast }-%
\mathbf{x}\right\| $.

Conversely, suppose that $\mathbf{F}$\ in (\ref{NSCS30}) is locally
Lipschitz for a common constant $C$, for all $\mathbf{s}:\ \left| \mathbf{s}%
\right| \leq s_{0}$. Let some points $\mathbf{x}$, $\mathbf{x}_{\ast }$ $\in
\Omega _{F}$ be joined by a continuous curve $\gamma $. In view of (\ref
{NSCS46}), the curve $\gamma $ can be taken such that $L\left( \gamma
\right) \leq \Lambda _{\Omega _{F}}\left( \mathbf{x},\mathbf{x}_{\ast
}\right) +\varepsilon $ for an arbitrary $\varepsilon >0$. Given any point $%
\mathbf{z}\in \gamma $, there is a ball $B_{\mathbf{z}}\subset \Omega _{F}$.
The balls $\left\{ B_{\mathbf{z}}\right\} $ form an open cover of $\gamma $.
Since the mapping $\gamma :\left[ 0,1\right] \subset \mathbb{R}\rightarrow 
\mathbb{R}^{N}$ is continuous, the curve $\gamma $ is compact \cite[p. 94]
{Kolmogorov and Fomin 1970}. Hence, by the compactness of $\gamma $, $%
\left\{ B_{\mathbf{z}}\right\} $ has a finite subcover consisting of balls $%
B_{\mathbf{x}}=B_{\mathbf{z}_{1}}$, $B_{\mathbf{z}_{2}}$, $\ldots $, $B_{%
\mathbf{z}_{K}}=B_{\mathbf{x}_{\ast }}$. Since $\mathbf{F}$\ is locally
Lipschitz, we find 
\begin{equation}
\left\| \mathbf{F}\left( \mathbf{z}_{k+1}\right) -\mathbf{F}\left( \mathbf{z}%
_{k}\right) \right\| \leq C\left\| \mathbf{z}_{k+1}-\mathbf{z}_{k}\right\|
,\quad k=1,2,\ldots K-1,\ \forall \ \mathbf{s}:\ \left| \mathbf{s}\right|
\leq s_{0}.  \label{NDS30}
\end{equation}
Then, by virtue of (\ref{NDS30}), we find 
\begin{equation*}
\left\| \mathbf{F}\left( \mathbf{x}_{\ast }\right) -\mathbf{F}\left( \mathbf{%
x}\right) \right\| =\left\| \sum_{k}\left[ \mathbf{F}\left( \mathbf{z}%
_{k+1}\right) -\mathbf{F}\left( \mathbf{z}_{k}\right) \right] \right\| \leq
C\sum_{k}\left\| \mathbf{z}_{k+1}-\mathbf{z}_{k}\right\| \leq
\end{equation*}
\begin{equation}
CL\left( \gamma \right) \leq C\Lambda _{\Omega _{F}}\left( \mathbf{x},%
\mathbf{x}_{\ast }\right) +\varepsilon C,\quad \forall \ \mathbf{s}:\ \left| 
\mathbf{s}\right| \leq s_{0}.  \label{NDS40}
\end{equation}
By letting $\varepsilon \rightarrow 0$, we find that (\ref{NDS20}) holds.
\end{proof}

\label{aaa0}Let us find the necessary and sufficient conditions for the
stability of Scheme (\ref{NSCS30}). Let $W^{1,\infty }\left( \Omega
_{F}\right) $ denote the Sobolev space, and let $\mathbf{F\equiv }\left\{
F_{1}\right. ,$ $F_{2},$ $\ldots ,$ $\left. F_{N}\right\} ^{T}$ in (\ref
{NSCS30}). Then, $F_{i}$, $i=1,2,\ldots ,N$, (and, hence, $\mathbf{F}$) is
locally Lipschitz (in the sense of having representatives) \emph{iff}\ $%
F_{i}\in W^{1,\infty }\left( \Omega _{F}\right) $ (see, e.g., \cite[Theorem
4.1]{Heinonen Juha 2005}). Let $\nabla F_{i}$ denote the distributional
gradient of $F_{i}$, and let $\delta \mathbf{F,}$ $\delta \mathbf{x}$ $%
\mathbf{\in }$ $\mathbb{R}^{N}$ denote variations. The following equality 
\begin{equation}
\delta \mathbf{F}=\mathbf{F}^{\prime }\cdot \delta \mathbf{x},\quad \mathbf{F%
}^{\prime }=\left\{ \nabla F_{1}\right. ,\nabla F_{2},\ldots ,\left. \nabla
F_{N}\right\} ^{T},  \label{NDS55}
\end{equation}
will be viewed as the scheme in variations for (\ref{NSCS30}).

\begin{lemma}
\label{LinSchStab}Linear Scheme (\ref{NDS55}) will be stable \emph{iff}\
there exist positive $s_{0}$, $C=const$ such that 
\begin{equation}
\left\| \mathbf{F}^{\prime }\right\| \leq C=const,\quad \forall \ \mathbf{x}%
\in \Omega _{F},\quad \forall \ \mathbf{s}:\ \left| \mathbf{s}\right| \leq
s_{0}.  \label{NDS65}
\end{equation}
\end{lemma}

\begin{proof}
The sufficiency is obvious. Actually, by virtue of (\ref{NDS65}), we find
that $\left\| \delta \mathbf{F}\right\| =\left\| \mathbf{F}^{\prime }\cdot
\delta \mathbf{x}\right\| \leq \left\| \mathbf{F}^{\prime }\right\| \left\|
\delta \mathbf{x}\right\| \leq C\left\| \delta \mathbf{x}\right\| $, i.e. 
\begin{equation}
\left\| \delta \mathbf{F}\right\| \leq C\left\| \delta \mathbf{x}\right\|
,\quad \forall \ \mathbf{x}\in \Omega _{F},\quad \forall \ \mathbf{s}:\
\left| \mathbf{s}\right| \leq s_{0}.  \label{NDS75}
\end{equation}
Conversely, suppose that (\ref{NDS75}) is valid. Then, in view of 
\cite[Theorem 2, p. 224]{Kolmogorov and Fomin 1970}, we write 
\begin{equation}
\left\| \mathbf{F}^{\prime }\right\| =\underset{\left\| \delta \mathbf{x}%
\right\| \neq 0}{\sup }\frac{\left\| \mathbf{F}^{\prime }\cdot \delta 
\mathbf{x}\right\| }{\left\| \delta \mathbf{x}\right\| }=\underset{\left\|
\delta \mathbf{x}\right\| \neq 0}{\sup }\frac{\left\| \delta \mathbf{F}%
\right\| }{\left\| \delta \mathbf{x}\right\| }\leq \underset{\left\| \delta 
\mathbf{x}\right\| \neq 0}{\sup }\frac{C\left\| \delta \mathbf{x}\right\| }{%
\left\| \delta \mathbf{x}\right\| }=C.  \label{NDS80}
\end{equation}
Hence, (\ref{NDS65}) holds, in view of (\ref{NDS80})
\end{proof}

\begin{theorem}
\label{iffStab}Consider Scheme (\ref{NSCS30}). Let the path-connected $%
\Omega _{F}$ be open, $\mathbf{F\equiv }\left\{ F_{1}\right. ,$ $F_{2},$ $%
\ldots ,$ $\left. F_{N}\right\} ^{T}$ be bounded, and let $\left\| \mathbf{F}%
^{\prime }\right\| \equiv \left\| \mathbf{f}_{F}\right\| $, $\mathbf{f}%
_{F}\equiv \left\{ \left\| \nabla F_{1}\right\| _{\infty }\right. $ $,$ $%
\left\| \nabla F_{2}\right\| _{\infty },$ $\ldots ,$ $\left. \left\| \nabla
F_{N}\right\| _{\infty }\right\} ^{T}$, $\nabla F_{i}$, $i=1,2,\ldots ,N$,
denote the distributional gradient of $F_{i}$. Then, Scheme (\ref{NSCS30})
will be stable \emph{iff}\ its scheme in variations, (\ref{NDS55}), will be
stable.
\end{theorem}

\begin{proof}
The proof is trivial. Actually, Scheme (\ref{NSCS30}) is stable $%
\Longleftrightarrow $ $\mathbf{F}$\ is locally Lipschitz for a common
constant $C$ (Lemma \ref{LipschitzStability}) $\Longleftrightarrow $ $%
F_{i}\in W^{1,\infty }\left( \Omega _{F}\right) $ (see \cite[Theorem 4.1]
{Heinonen Juha 2005}) $\Longleftrightarrow $ (\ref{NDS65}) holds $%
\Longleftrightarrow $\ Scheme in variations, (\ref{NDS55}), is stable (Lemma 
\ref{LinSchStab}).
\end{proof}

Notice, if $\mathbf{F}$ in (\ref{NSCS30}) is Gateaux-differentiable, then $%
\nabla F_{i}$\ (see Theorem \ref{iffStab})\ denotes the classical gradient
of $F_{i}$, and, hence, it may be taken that $\mathbf{f}_{F}=\left\{ \left\|
\nabla F_{1}\right\| \right. $ $,$ $\left\| \nabla F_{2}\right\| ,$ $\ldots
, $ $\left. \left\| \nabla F_{N}\right\| \right\} ^{T}$, see also 
\cite[Theorem 3]{Borisov and Mond 2008}.

\section{Monotone $C^{1}$ piecewise cubics in construction of central schemes%
\label{COS}}

In this section we consider some theoretical aspects for high-order
interpolation and employment of monotone $C^{1}$ piecewise cubics (e.g., 
\cite{Fritsch and Carlson 1980}, \cite{Kocic and Milovanovic 1997}) in
construction of monotone central schemes. We will consider explicit schemes
on a uniform grid with time step $\Delta t$ and spatial mesh size $\Delta x$%
, as applied to the following hyperbolic 1-D equation 
\begin{equation}
\frac{\partial \mathbf{u}}{\partial t}+\frac{\partial }{\partial x}\mathbf{f}%
\left( \mathbf{u}\right) =0,\ t_{n}<t\leq t_{n+1}\equiv t_{n}+\Delta t,\quad 
\mathbf{u}\left( x,t_{n}\right) =\mathbf{u}^{n}\left( x\right) ,  \label{C10}
\end{equation}
Using the central differencing, we write 
\begin{equation}
\left. \frac{\partial \mathbf{u}}{\partial t}\right| _{t=t_{n+0.25},\
x=x_{i+0.5}}=\frac{\mathbf{u}_{i+0.5}^{n+0.5}-\mathbf{u}_{i+0.5}^{n}}{%
0.5\Delta t}+O\left( \left( \Delta t\right) ^{2}\right) ,  \label{C24}
\end{equation}
\begin{equation}
\left. \frac{\partial \mathbf{f}}{\partial x}\right| _{t=t_{n+0.25},\
x=x_{i+0.5}}=\frac{\mathbf{f}_{i+1}^{n+0.25}-\mathbf{f}_{i}^{n+0.25}}{\Delta
x}+O\left( \left( \Delta x\right) ^{2}\right) .  \label{C25}
\end{equation}
By virtue of (\ref{C24})-(\ref{C25}) we approximate (\ref{C10}) on the cell $%
\left[ x_{i},x_{i+1}\right] \times \left[ t_{n},t_{n+0.5}\right] $ by the
following difference equation 
\begin{equation}
\mathbf{v}_{i+0.5}^{n+0.5}=\mathbf{v}_{i+0.5}^{n}-\frac{\Delta t}{2\Delta x}%
\left( \mathbf{g}_{i+1}^{n+0.25}-\mathbf{g}_{i}^{n+0.25}\right) .
\label{C30}
\end{equation}
As usual, the mathematical treatment for the second step (i.e., on the cell $%
\left[ x_{i-0.5},x_{i+0.5}\right] \times \left[ t_{n+0.5},t_{n+1}\right] $)
of a staggered scheme will, in general, not be included in the text, because
it is quite similar to the one for the first step.

Considering that (\ref{C30}) approximates (\ref{C10}) with the accuracy $%
O(\left( \Delta x\right) ^{2}+\left( \Delta t\right) ^{2})$, the next
problem is to approximate $\mathbf{v}_{i+0.5}^{n}$ and $\mathbf{g}%
_{i}^{n+0.25}$ in such a way as to retain the accuracy of the approximation.
For instance, the following approximations 
\begin{equation}
\mathbf{v}_{i+0.5}^{n}=0.5\left( \mathbf{v}_{i}^{n}+\mathbf{v}%
_{i+1}^{n}\right) +O\left( \left( \Delta x\right) ^{2}\right) ,\quad \mathbf{%
g}_{i}^{n+0.25}=\mathbf{f}\left( \mathbf{v}_{i}^{n}\right) +O\left( \Delta
t\right) ,  \label{C50}
\end{equation}
leads to the staggered form of the famed LxF scheme that is of the
first-order approximation (see, e.g., \cite[p. 170]{Godlewski and Raviart
1996}). One way to obtain a higher-order scheme is to use a higher order
interpolation. At the same time it is required of the interpolant to be
monotonicity preserving. Notice, the classic cubic spline does not possess
such a property (see Figure \ref{Fritsch}a). Let us consider the problem of
high-order interpolation of $\mathbf{v}_{i+0.5}^{n}$ in (\ref{C30}) with
closer inspection

Let $\mathbf{p}=\mathbf{p}\left( x\right) \equiv \left\{ p^{1}\left(
x\right) ,\ldots ,p^{k}\left( x\right) ,\ldots ,p^{m}\left( x\right)
\right\} ^{T}$ be a component-wise monotone $C^{1}$ piecewise cubic
interpolant (e.g., \cite{Fritsch and Carlson 1980}, \cite{Kocic and
Milovanovic 1997}), and let 
\begin{equation*}
\mathbf{p}_{i}=\mathbf{p}\left( x_{i}\right) ,\quad \mathbf{p}_{i}^{\prime }=%
\mathbf{p}^{\prime }\left( x_{i}\right) ,\quad \Delta \mathbf{p}_{i}=\mathbf{%
p}_{i+1}-\mathbf{p}_{i},
\end{equation*}
\begin{equation}
\mathbf{p}_{i}^{\prime }=\mathbb{A}_{i}\cdot \frac{\Delta \mathbf{p}_{i}}{%
\Delta x},\quad \mathbf{p}_{i+1}^{\prime }=\mathbb{B}_{i}\cdot \frac{\Delta 
\mathbf{p}_{i}}{\Delta x},  \label{C80}
\end{equation}
where $\mathbf{p}_{i}^{\prime }$ denotes the derivative of the interpolant
at $x=x_{i}$. The diagonal matrices $\mathbb{A}_{i}$ and $\mathbb{B}_{i}$\
in (\ref{C80})\ are defined as follows 
\begin{equation}
\mathbb{A}_{i}=diag\left\{ \alpha _{i}^{1},\alpha _{i}^{2},\ldots ,\alpha
_{i}^{m}\right\} ,\ \mathbb{B}_{i}=diag\left\{ \beta _{i}^{1},\beta
_{i}^{2},\ldots ,\beta _{i}^{m}\right\} .  \label{C85}
\end{equation}
The cubic interpolant, $\mathbf{p}=\mathbf{p}\left( x\right) $, is
component-wise monotone on $\left[ x_{i},x_{i+1}\right] $ \emph{iff} one of
the following conditions (e.g., \cite{Fritsch and Carlson 1980}, \cite{Kocic
and Milovanovic 1997}) is satisfied: 
\begin{equation}
\left( \alpha _{i}^{k}-1\right) ^{2}+\left( \alpha _{i}^{k}-1\right) \left(
\beta _{i}^{k}-1\right) +\left( \beta _{i}^{k}-1\right) ^{2}-3\left( \alpha
_{i}^{k}+\beta _{i}^{k}-2\right) \leq 0,  \label{C90}
\end{equation}
\begin{equation}
\alpha _{i}^{k}+\beta _{i}^{k}\leq 3,\quad \alpha _{i}^{k}\geq 0,\ \beta
_{i}^{k}\geq 0,\quad \forall i,k.  \label{C100}
\end{equation}
As reported in \cite{Kocic and Milovanovic 1997}, the necessary and
sufficient conditions for monotonicity of a $C^{1}$ piecewise cubic
interpolant originally given by Ferguson and Miller (1969), and
independently, by Fritsch and Carlson \cite{Fritsch and Carlson 1980}. The
region of monotonicity is shown in Figure \ref{Fritsch}b. The results of
implementing a monotone $C^{1}$ piecewise cubic interpolation when compared
with the classic cubic spline interpolation, are depicted in Figure \ref
{Fritsch}a. We note (Figure \ref{Fritsch}a) that the constructed function
produces monotone interpolation and this function coincides with the classic
cubic spline at some sections where the classic cubic spline is monotone.

\begin{figure}[h]
\centerline{\includegraphics[width=11.50cm,height=4.50cm]{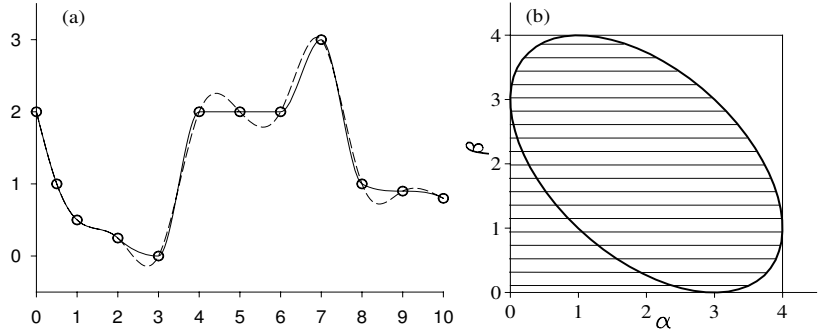}}
\caption{Monotone piecewise cubic interpolation. (a) Interpolation of a 1-D
tabulated function. Circles: prescribed tabulated values; Dashed line:
classic cubic spline; Solid line: monotone piecewise cubic. (b) Necessary
and sufficient conditions for monotonicity. Horizontal hatching: region of
monotonicity; Unshaded: cubic is non-monotone. }
\label{Fritsch}
\end{figure}

Using the cubic segment of the $C^{1}$ piecewise cubic interpolant, $\mathbf{%
p}=\mathbf{p}\left( x\right) $, (see, e.g., \cite{Fritsch and Carlson 1980}, 
\cite{Kocic and Milovanovic 1997}) for $x\in \left[ x_{i},x_{i+1}\right] $,
we obtain the following interpolation formula 
\begin{equation}
\mathbf{p}_{i+0.5}=0.5\left( \mathbf{p}_{i}+\mathbf{p}_{i+1}\right) -\frac{%
\Delta x}{8}\left( \mathbf{p}_{i+1}^{\prime }-\mathbf{p}_{i}^{\prime
}\right) +O\left( \left( \Delta x\right) ^{r}\right) .  \label{C110}
\end{equation}
If $\mathbf{p}\left( x\right) $ has a continuous fourth derivative, then $%
r=4 $ in (\ref{C110}), see e.g. \cite[p. 111]{Kahaner et al. 1989}. However,
the exact value of $\mathbf{p}_{i}^{\prime }$ in (\ref{C110}) is, in
general, unknown, and hence to construct numerical schemes, employing
formulae similar to (\ref{C110}), the value of derivatives $\mathbf{p}%
_{i}^{\prime }$ must be estimated.

Using (\ref{C110}) and the second formula in (\ref{C50}) we obtain from (\ref
{C30}) the following scheme 
\begin{equation}
\mathbf{v}_{i+0.5}^{n+0.5}=0.5\left( \mathbf{v}_{i}^{n}+\mathbf{v}%
_{i+1}^{n}\right) -\frac{\Delta x}{8}\left( \mathbf{d}_{i+1}^{n}-\mathbf{d}%
_{i}^{n}\right) -\frac{\Delta t}{2}\frac{\mathbf{f}\left( \mathbf{v}%
_{i+1}^{n}\right) -\mathbf{f}\left( \mathbf{v}_{i}^{n}\right) }{\Delta x},
\label{C120}
\end{equation}
where $\mathbf{d}_{i}^{n}$ denotes the derivative of the interpolant at $%
x=x_{i}$. In view of (\ref{C110}) and the second formula in (\ref{C50}), the
local truncation error \cite[p. 142]{LeVeque 2002}, $\psi $, on a
sufficiently smooth solution $\mathbf{u}(x,t)$ to (\ref{C10}) is found to be 
\begin{equation}
\psi =O\left( \Delta t\right) +O\left( \frac{\left( \Delta x\right) ^{r}}{%
\Delta t}\right) +O\left( \left( \Delta t\right) ^{2}+\left( \Delta x\right)
^{2}\right) .  \label{C130}
\end{equation}
In view of (\ref{C130}) we conclude that the scheme (\ref{C120}) generates a
conditional approximation, because it approximates (\ref{C10}) only if $%
\left( \Delta x\right) ^{r}\diagup \Delta t\rightarrow 0$ as $\Delta
x\rightarrow 0$ and $\Delta t\rightarrow 0$. Let $\mathbf{d}_{i}^{n}$\ be
approximated with the accuracy $O\left( \left( \Delta x\right) ^{s}\right) $%
, then the value of $r$ in (\ref{C130}) can be calculated (see Section \ref
{Appendix1}, Proposition \ref{Approximate derivative}) by the following
formula 
\begin{equation}
r=\min \left( 4,s+1\right) .  \label{C140}
\end{equation}
Interestingly, since (\ref{C120}) provides the conditional approximation,
the order of accuracy depends on the pathway taken by $\Delta x$ and $\Delta
t$ as $\Delta x\rightarrow 0$ and $\Delta t\rightarrow 0$. Actually, there
exists a pathway such that $\Delta t$ is proportional to $\left( \Delta
x\right) ^{\mu }$ and the CFL condition is fulfilled provided $\mu \geq 1$
and $\Delta x\leq \Delta x_{0}$, where $\Delta x_{0}$ is a positive value.
If we take $\mu =1$ and $s\geq 1$, then we obtain from (\ref{C130}) that the
scheme (\ref{C120}) is of the first-order. If $\mu =2$ and $s\geq 3$, then (%
\ref{C120}) is of the second-order. However, if $\mu =2$ and $s=2$, then, in
view of (\ref{C130}) and (\ref{C140}), the scheme (\ref{C120}) is of the
first-order. Moreover, under $\mu =2$ and $s=2$, the scheme will be of the
first-order even if $\mathbf{g}_{i}^{n+0.25}$ in (\ref{C30}) will be
approximated with the accuracy$\ O(\left( \Delta t\right) ^{2})$. It seems
likely that Example 6 in \cite{Kurganov and Tadmor 2000} can be seen as an
illustration of the last assertion. The Nessyahu-Tadmor (NT) scheme with the
second-order approximation of $\mathbf{d}_{i}^{n}$ is used \cite{Kurganov
and Tadmor 2000} to solve a Burgers-type equation. Since $\Delta t$ $=$ $%
O(\left( \Delta x\right) ^{2})$ \cite{Kurganov and Tadmor 2000}, the NT
scheme is of the first-order, and hence it can be the main reason for the
scheme to exhibit the smeared discontinuity computed in \cite[Fig. 6.22]
{Kurganov and Tadmor 2000}.

The approximation of derivatives $\mathbf{p}_{i}^{\prime }$ can be done by
the following three steps \cite{Fritsch and Carlson 1980}: (i) an
initialization of the derivatives $\mathbf{p}_{i}^{\prime }$; (ii) the
choice of subregion of monotonicity; (iii) modification of the initialized
derivatives $\mathbf{p}_{i}^{\prime }$ to produce a monotone interpolant.

The matter of initialization of the derivatives is the most subtle issue of
this algorithm. Actually, the approximation of $\mathbf{p}_{i}^{\prime }$
must, in general, be done with accuracy $O(\left( \Delta x\right) ^{3})$ to
obtain the second-order scheme when $\Delta t$ is proportional to $\left(
\Delta x\right) ^{2}$, inasmuch as central schemes generate a conditional
approximation. Thus, using the two-point or the three-point (centered)
difference formula (e.g. \cite{Kocic and Milovanovic 1997}, \cite{Pareschi
Lorenzo 2001}) we obtain, in general, the first-order scheme. The so called
limiter functions \cite{Kocic and Milovanovic 1997} lead, in general, to a
low-order scheme as these limiters are often $O(\Delta x)$ or $O(\left(
\Delta x\right) ^{2})$ accurate. Performing the initialization of the
derivatives $\mathbf{p}_{i}^{\prime }$ in the interpolation formula (\ref
{C110}) by the classic cubic spline interpolation \cite{Press William 1988},
we obtain the approximation, which is $O(\left( \Delta x\right) ^{3})$
accurate (e.g., \cite{Kahaner et al. 1989}, \cite{Kocic and Milovanovic 1997}%
), and hence, in general, the second-order scheme. The same accuracy, $%
O(\left( \Delta x\right) ^{3})$, can be achieved by using the four-point
approximation \cite{Kocic and Milovanovic 1997}. However, the efficiency of
the algorithm based on the classic cubic spline interpolation is comparable
with the one based on the four-point approximation, as the number of
multiplications and divisions (as well as additions and subtractions) per
one node is approximately the same for both algorithms. We will use the
classic cubic spline interpolation for the initialization of the derivatives 
$\mathbf{P}_{i}^{\prime }$ in the interpolation formula (\ref{C110}), as it
is based on the tridiagonal algorithm, which is `the rare case of an
algorithm that, in practice, is more robust than theory says it should be' 
\cite{Press William 1988}.

Obviously, for each interval $\left[ x_{i},x_{i+1}\right] $ in which the
initialized derivatives $\mathbf{p}_{i}^{\prime }$, $\mathbf{p}%
_{i+1}^{\prime }$ such that at least one point ($\alpha _{i}^{k}$, $\beta
_{i}^{k}$) does not belong to the region of monotonicity (\ref{C90})-(\ref
{C100}), the derivatives $\mathbf{p}_{i}^{\prime }$, $\mathbf{p}%
_{i+1}^{\prime }$ must be modified to $\widetilde{\mathbf{p}}_{i}^{\prime }$%
, $\widetilde{\mathbf{p}}_{i+1}^{\prime }$ such that the point ($\widetilde{%
\alpha }_{i}^{k}$, $\widetilde{\beta }_{i}^{k}$) will be in the region of
monotonicity. 
The modification of the initialized derivatives, would be much simplified if
we take a square as a subregion of monotonicity. In connection with this, we
will make use the subregions of monotonicity represented in the following
form: 
\begin{equation}
0\leq \alpha _{i}^{k}\leq 4\aleph ,\quad 0\leq \beta _{i}^{k}\leq 4\aleph
,\quad \forall i,k,  \label{CA180}
\end{equation}
where $\aleph $ is a monotonicity parameter. 
Obviously, the condition (\ref{CA180}) is sufficient for the monotonicity
(see Figure \ref{Fritsch}b) provided that $0\leq $ $\aleph $ $\leq 0.75$.

Let us now find necessary and sufficient conditions for (\ref{C110}) to be
monotonicity preserving. By virtue of (\ref{C80}), the interpolation formula
(\ref{C110}) can be rewritten to read 
\begin{equation}
\mathbf{p}_{i+0.5}=\left( 0.5\mathbf{I}+\frac{\mathbb{B}_{i}-\mathbb{A}_{i}}{%
8}\right) \cdot \mathbf{p}_{i}+\left( 0.5\mathbf{I}-\frac{\mathbb{B}_{i}-%
\mathbb{A}_{i}}{8}\right) \cdot \mathbf{p}_{i+1}.  \label{CA185}
\end{equation}
The coefficients of (\ref{CA185}) will be non-negative \emph{iff} $\left|
\beta _{i}-\alpha _{i}\right| \leq 4$. Hence (\ref{C110}) will be
monotonicity preserving \emph{iff} (\ref{CA180}) will be valid provided $%
0\leq $ $\aleph $ $\leq 1$. Notice, there is no any contradiction between
the sufficient conditions, (\ref{CA180}) provided $0\leq $ $\aleph $ $\leq
0.75$, for the interpolant, $\mathbf{p}=\mathbf{p}\left( x\right) $,\ to be
monotone through the interval $\left[ x_{i},x_{i+1}\right] $, and the
necessary and sufficient conditions, (\ref{CA180}) provided $0\leq $ $\aleph 
$ $\leq 1$, for the scheme (\ref{CA185}) to be monotonicity preserving. In
the latter case the interpolant, $\mathbf{p}=\mathbf{p}\left( x\right) $,\
may, in general, be non-monotone, however at the point $i+0.5$ the value of
an arbitrary component of $\mathbf{p}_{i+0.5}$ will be between the
corresponding components of $\mathbf{p}_{i}$ and $\mathbf{p}_{i+1}$.

To fulfill the conditions of monotonicity (\ref{CA180}), the modification of
derivatives $\mathbf{p}_{i}^{\prime }=\left\{ p_{i}^{\prime 1},p_{i}^{\prime
2},\ldots ,p_{i}^{\prime m}\right\} $ can be done by the following algorithm
suggested, in fact, by Fritsch and Carlson \cite{Fritsch and Carlson 1980}
(see also \cite{Kocic and Milovanovic 1997}): 
\begin{equation}
S_{i}^{k}:=4\aleph \min \func{mod}(\Delta _{i-1}^{k},\Delta _{i}^{k}),\quad 
\widetilde{p}_{i}^{\prime k}:=\min \func{mod}(p_{i}^{\prime
k},S_{i}^{k}),\quad \aleph =const,  \label{REMa10}
\end{equation}
where $\Delta _{i}^{k}=\left( p_{i+1}^{k}-p_{i}^{k}\right) \diagup \Delta x$%
, the function $\min \func{mod}(x,y)$ is defined (e.g., \cite{Kocic and
Milovanovic 1997}, \cite{Kurganov and Tadmor 2000}, \cite{Morton 2001}, \cite
{Pareschi Lorenzo 2001}, \cite{Serna and Marquina 2005}) as follows 
\begin{equation}
\min \func{mod}(x,y)\equiv \frac{1}{2}\left[ sgn(x)+sgn(y)\right] \min
\left( \left| x\right| ,\left| y\right| \right) .  \label{REMa20}
\end{equation}

Let us note that instead of point values, $\mathbf{v}_{i+0.5}^{n}$, employed
in the construction of the scheme (\ref{C30}), it can be used the cell
averages (e.g., \cite{Balaguer and Conde 2005}, \cite{Kurganov and Tadmor
2000}, \cite{LeVeque 2002}) calculated on the basis of the monotone $C^{1}$
piecewise cubics. In such a case we obtain, instead of (\ref{C110}), the
following interpolation formula 
\begin{equation}
\mathbf{p}_{i+0.5}=0.5\left( \mathbf{p}_{i}+\mathbf{p}_{i+1}\right)
-\varkappa \frac{\Delta x}{8}\left( \mathbf{p}_{i+1}^{\prime }-\mathbf{p}%
_{i}^{\prime }\right) ,  \label{CA200}
\end{equation}
where $\varkappa =2\diagup 3$. The region of monotonicity in this case will
also be 
\begin{equation}
0\leq \mathbb{A}_{i}\leq 4\aleph \mathbf{I},\ 0\leq \mathbb{B}_{i}\leq
4\aleph \mathbf{I},\quad 0\leq \aleph \leq 1,\quad \forall i.  \label{CA203}
\end{equation}
Notice, the interpolation formula (\ref{CA200}) coincides with (\ref{C110})
under $\varkappa =1$. Thus, in view of the interpolation formula (\ref{CA200}%
), the staggered scheme (\ref{C30}) is written to read 
\begin{equation}
\mathbf{v}_{i+0.5}^{n+0.5}=0.5\left( \mathbf{v}_{i+1}^{n}+\mathbf{v}%
_{i}^{n}\right) -\varkappa \frac{\Delta x}{8}\left( \mathbf{d}_{i+1}^{n}-%
\mathbf{d}_{i}^{n}\right) -\frac{\Delta t}{2}\frac{\mathbf{f}\left( \mathbf{v%
}_{i+1}^{n}\right) -\mathbf{f}\left( \mathbf{v}_{i}^{n}\right) }{\Delta x},
\label{CA210}
\end{equation}
where $\mathbf{d}_{i}^{n}$ denotes the derivative of the interpolant at $%
x=x_{i}$, the range of values for the parameter $\varkappa $ is the segment $%
0\leq \varkappa \leq 1$. If $\varkappa =1$ (or $\varkappa =0$), then Scheme (%
\ref{CA210}) coincides with the scheme (\ref{C120}) (or with the LxF scheme,
respectively). As it was shown above, the scheme (\ref{CA210}) is of the
first order provided $\Delta t$ $=$ $O\left( \Delta x\right) $. The central
scheme (\ref{CA210}), approximating the 1-D equation (\ref{C10}) with the
first order, will be abbreviated to as COS1.

\section{Construction of central schemes\label{COSN}}

We will consider explicit schemes on a uniform grid with time step $\Delta t$
and spatial mesh size $\Delta x$. In view of the CFL condition \cite{LeVeque
2002}, we assume for the explicit schemes, that $\Delta t=O\left( \Delta
x\right) $. Moreover, we will also assume that $\Delta x=O\left( \Delta
t\right) $, since a central scheme generates a conditional approximation to
Eq. (\ref{C10}) (see Section \ref{COS}). In such a case, the following
inequalities will be valid, for sufficiently small $\Delta t$ and $\Delta x$%
, 
\begin{equation}
\nu _{0}\Delta t\leq \Delta x\leq \mu _{0}\Delta t,\quad \nu _{0},\mu
_{0}=const,\ 0<\nu _{0}\leq \mu _{0}.  \label{CA05}
\end{equation}
Notice, for hyperbolic problems it is often assumed that $\Delta t$ and $%
\Delta x$ are related in a fixed manner (e.g., \cite[p. 140]{LeVeque 2002}, 
\cite[p. 120]{Richtmyer and Morton 1967}), i.e. it is assumed that $\Delta t$
and $\Delta x$ fulfill a more strong condition than (\ref{CA05}).

Scheme (\ref{C120}) is of the first-order, $O(\Delta t+\left( \Delta
x\right) ^{2})$, and non-oscillatory LxF scheme is of the first-order, $%
O(\Delta t+\Delta x)$. Let us demonstrate that (\ref{C120}) is, in fact, LxF
scheme with a negative numerical viscosity added to obtain a higher order
approximation to Eq. (\ref{C10}) with respect to $x$. We rewrite Scheme (\ref
{CA210}) to read 
\begin{equation*}
\frac{\mathbf{v}_{i+0.5}^{n+0.5}-\mathbf{v}_{i+0.5}^{n}}{0.5\Delta t}+\frac{%
\mathbf{f}\left( \mathbf{v}_{i+1}^{n}\right) -\mathbf{f}\left( \mathbf{v}%
_{i}^{n}\right) }{\Delta x}=
\end{equation*}
\begin{equation}
\frac{\Delta x^{2}}{\Delta t}\frac{\mathbf{v}_{i}^{n}-2\mathbf{v}%
_{i+0.5}^{n}+\mathbf{v}_{i+1}^{n}}{\Delta x^{2}}-\varkappa \frac{\Delta x^{2}%
}{4\Delta t}\frac{\mathbf{d}_{i+1}^{n}-\mathbf{d}_{i}^{n}}{\Delta x}.
\label{INA30}
\end{equation}
Notice, the second term in the right-hand side of (\ref{INA30}) is, in fact,
the negative numerical viscosity. Without this term ($\varkappa =0$), Scheme
(\ref{INA30}) would be LxF scheme. As it is demonstrated in Section \ref
{Exemplification and discussion}, Scheme (\ref{C120}) can exhibit spurious
oscillations in contrast to LxF scheme. Interestingly, there is a
possibility to improve Scheme (\ref{CA210}) by introducing an additional
positive numerical viscosity such that the scheme's order of accuracy would
increase up to $O((\Delta t)^{2}+\left( \Delta x\right) ^{2})$. Let us
approximate $\mathbf{v}_{i+0.5}^{n}$\ and $\mathbf{g}_{i}^{n+0.125}$ in (\ref
{C30}) with the accuracy $O(\left( \Delta x\right) ^{2}+\left( \Delta
t\right) ^{2})$. Using Taylor series expansion, we write 
\begin{equation}
\mathbf{g}_{i}^{n+0.25}=\mathbf{f}\left( \mathbf{v}_{i}^{n}\right) +\left. 
\frac{\partial \mathbf{f}\left( \mathbf{v}_{i}^{n}\right) }{\partial t}%
\right| _{t=t_{n}}\frac{\Delta t}{4}+O\left( \Delta t^{2}\right) .
\label{SA10}
\end{equation}
By virtue of the PDE system, (\ref{C10}), we find 
\begin{equation}
\frac{\partial \mathbf{f}}{\partial t}=\frac{\partial \mathbf{f}}{\partial 
\mathbf{u}}\cdot \frac{\partial \mathbf{u}}{\partial t}=-\frac{\partial 
\mathbf{f}}{\partial \mathbf{u}}\cdot \frac{\partial \mathbf{f}}{\partial 
\mathbf{u}}\cdot \frac{\partial \mathbf{u}}{\partial x}=-\left( \frac{%
\partial \mathbf{f}}{\partial \mathbf{u}}\right) ^{2}\cdot \frac{\partial 
\mathbf{u}}{\partial x}.  \label{SA20}
\end{equation}
Using the interpolation formula (\ref{CA200}) and the formulae (\ref{SA10})-(%
\ref{SA20}), we obtain from (\ref{C30}) the following second order central
scheme 
\begin{equation*}
\mathbf{v}_{i+0.5}^{n+0.5}=0.5\left( \mathbf{v}_{i+1}^{n}+\mathbf{v}%
_{i}^{n}\right) -\varkappa \frac{\Delta x}{8}\left( \mathbf{d}_{i+1}^{n}-%
\mathbf{d}_{i}^{n}\right) -\frac{\Delta t}{2}\frac{\mathbf{f}\left( \mathbf{v%
}_{i+1}^{n}\right) -\mathbf{f}\left( \mathbf{v}_{i}^{n}\right) }{\Delta x}+
\end{equation*}
\begin{equation}
\xi \frac{\left( \Delta t\right) ^{2}}{8\Delta x}\left[ \left( \mathbf{A}%
_{i+1}^{n}\right) ^{2}\cdot \mathbf{d}_{i+1}^{n}-\left( \mathbf{A}%
_{i}^{n}\right) ^{2}\cdot \mathbf{d}_{i}^{n}\right] ,\quad \mathbf{A\equiv }%
\frac{\partial \mathbf{f}}{\partial \mathbf{u}},  \label{SA30}
\end{equation}
where $\xi $ is introduced by analogy with $\varkappa $ in (\ref{CA210}),
and hence $0\leq \xi \leq 1$. Scheme (\ref{SA30}) coincides with (\ref{CA210}%
) provided that $\xi =0$. Since $\mathbf{d}_{i}^{n}$ is the derivative of
the interpolant at $x=x_{i}$, the last term in right-hand side of (\ref{SA30}%
) can be seen as the non-negative numerical viscosity introduced into the
first order scheme (\ref{CA210}). Owing to this term, Scheme (\ref{SA30}) is 
$O(\left( \Delta x\right) ^{2}+\left( \Delta t\right) ^{2})$ accurate,
provided that $\xi =1$. Thus, we are dealing with the vanishing viscosity
method \cite{Godlewski and Raviart 1996}, \cite{LeVeque 2002} and, hence, in
view of \cite[Theorem 3.3]{Godlewski and Raviart 1996}, the scheme, (\ref
{SA30}), satisfies the entropy condition. The central scheme (\ref{SA30}),
approximating the 1-D equation (\ref{C10}) with the second order, will be
abbreviated to as COS2.

\subsection{Stability of the second-order scheme COS2}

In view of Theorem \ref{iffStab}, the stability of (\ref{SA30}) will be
investigated on the basis of its variational scheme.\ It is assumed that the
bounded operator $\mathbf{A}$ $(=\partial \mathbf{f}\left( \mathbf{u}\right)
\diagup \partial \mathbf{u)}$ in (\ref{C10}) is Fr\'{e}chet-differentiable
on the set $\Omega _{\mathbf{u}}$ $\subset $ $\mathbb{R}^{M}$, and its
derivative is bounded on $\Omega _{\mathbf{u}}$. Hence, the following
inequalities are valid 
\begin{equation}
\underset{\mathbf{u\in }\Omega _{\mathbf{u}}}{\sup }\left\| \mathbf{A}%
\right\| \leq \lambda _{\max }<\infty ,\quad \left\| \delta \mathbf{A}%
_{i}^{n}\right\| =\left\| \frac{\partial \mathbf{A}_{i}^{n}}{\partial 
\mathbf{v}_{i}^{n}}\cdot \delta \mathbf{v}_{i}^{n}\right\| \leq \alpha
_{A}\left\| \delta \mathbf{v}_{i}^{n}\right\| ,  \label{SB05}
\end{equation}
where $\lambda _{\max }$, $\alpha _{A}=const$. Considering that $\mathbf{v}%
_{i}^{n}$ in (\ref{SA30}) is Lipschitz-continuous, we write 
\begin{equation}
\left\| \mathbf{v}_{i}^{n}-\mathbf{v}_{i+1}^{n}\right\| \leq C_{v}\Delta
x,\quad C_{v}=const.  \label{SB40}
\end{equation}
By virtue of (\ref{C80}), the second term in right-hand side of (\ref{SA30})
can be written in the form 
\begin{equation}
\varkappa \frac{\Delta x}{8}\left( \mathbf{d}_{i+1}^{n}-\mathbf{d}%
_{i}^{n}\right) =\frac{\varkappa }{8}\left( \mathbb{B}_{i}^{n}-\mathbb{A}%
_{i}^{n}\right) \cdot \left( \mathbf{v}_{i+1}^{n}-\mathbf{v}_{i}^{n}\right) .
\label{SB10}
\end{equation}
Then, the variational scheme corresponding to (\ref{SA30}) is the following

\begin{equation*}
\delta \mathbf{v}_{i+0.5}^{n+0.5}=0.5\left( \delta \mathbf{v}_{i}^{n}+\delta 
\mathbf{v}_{i+1}^{n}\right) +\frac{\varkappa }{8}\left[ \left( \mathbf{v}%
_{i}^{n}-\mathbf{v}_{i+1}^{n}\right) ^{T}\cdot \delta \mathbb{D}_{i}^{n}%
\right] ^{T}+\frac{\varkappa }{8}\mathbb{D}_{i}^{n}\cdot \left( \delta 
\mathbf{v}_{i}^{n}-\delta \mathbf{v}_{i+1}^{n}\right)
\end{equation*}
\begin{equation*}
+\xi \frac{\left( \Delta t\right) ^{2}}{8\Delta x^{2}}\left\{ \left[ \delta
\left( \left( \mathbf{A}_{i+1}^{n}\right) ^{2}\cdot \mathbb{B}_{i}\right) %
\right] \cdot \left( \mathbf{v}_{i+1}^{n}-\mathbf{v}_{i}^{n}\right) -\left[
\delta \left( \left( \mathbf{A}_{i}^{n}\right) ^{2}\cdot \mathbb{A}%
_{i}\right) \right] \cdot \left( \mathbf{v}_{i+1}^{n}-\mathbf{v}%
_{i}^{n}\right) \right\} +
\end{equation*}
\begin{equation*}
\xi \frac{\left( \Delta t\right) ^{2}}{8\Delta x^{2}}\left[ \left( \mathbf{A}%
_{i+1}^{n}\right) ^{2}\cdot \mathbb{B}_{i}-\left( \mathbf{A}_{i}^{n}\right)
^{2}\cdot \mathbb{A}_{i}\right] \cdot \left( \delta \mathbf{v}%
_{i+1}^{n}-\delta \mathbf{v}_{i}^{n}\right) +
\end{equation*}
\begin{equation}
\frac{\Delta t}{2\Delta x}\left( \mathbf{A}_{i}^{n}\cdot \delta \mathbf{v}%
_{i}^{n}-\mathbf{A}_{i+1}^{n}\cdot \delta \mathbf{v}_{i+1}^{n}\right) ,
\label{SB20}
\end{equation}
where $\mathbb{D}_{i}^{n}=diag\left\{ D_{i,1}^{n},D_{i,2}^{n},\ldots
,D_{i,M}^{n}\right\} \equiv \mathbb{B}_{i}^{n}-\mathbb{A}_{i}^{n}$. By
virtue of (\ref{CA203}), we find that $-4\aleph \mathbf{I}\leq \mathbb{D}%
_{i}^{n}\leq 4\aleph \mathbf{I}$, and hence $-8\aleph \mathbf{I}\leq \delta 
\mathbb{D}_{i}^{n}\leq 8\aleph \mathbf{I}$. Thus, we may write that 
\begin{equation}
\left\| \delta \mathbb{D}_{i}^{n}\right\| \leq 8\aleph .  \label{SB30}
\end{equation}
By virtue of (\ref{CA05}), (\ref{SB30}) and (\ref{SB40}), and since $0\leq $ 
$\varkappa ,\aleph $ $\leq 1$, we find the following estimation for the
second term in right-hand side of (\ref{SB20}): 
\begin{equation}
\left\| \frac{\varkappa }{8}\left[ \left( \mathbf{v}_{i}^{n}-\mathbf{v}%
_{i+1}^{n}\right) ^{T}\cdot \delta \mathbb{D}_{i}^{n}\right] ^{T}\right\|
\leq \frac{\varkappa }{8}\left\| \mathbf{v}_{i}^{n}-\mathbf{v}%
_{i+1}^{n}\right\| \left\| \delta \mathbb{D}_{i}^{n}\right\| \leq \mu
_{0}C_{v}\Delta t.  \label{SB45}
\end{equation}
By virtue of (\ref{SB05}), (\ref{CA05}), (\ref{SB30}), (\ref{SB40}), and
since $0\leq $ $\xi $ $\leq 1$ and the CFL number $C_{r}=\Delta t\lambda
_{\max }\diagup \Delta x\leq 1$, we find the following estimation for the
fourth and fifth terms in right-hand side of (\ref{SB20}): 
\begin{equation}
\left\| \xi \frac{\left( \Delta t\right) ^{2}}{8\Delta x^{2}}\left[ \delta
\left( \left( \mathbf{A}_{i+1}^{n}\right) ^{2}\cdot \mathbb{B}_{i}\right) %
\right] \cdot \left( \mathbf{v}_{i+1}^{n}-\mathbf{v}_{i}^{n}\right) \right\|
\leq \beta _{A}\Delta t\left\| \delta \mathbf{v}_{i}^{n}\right\| ,\ \beta
_{A}=const,  \label{SB47}
\end{equation}
where $\beta _{A}$ depends on the other constants, namely, on $\alpha _{A}$, 
$\lambda _{\max }$, $C_{v}$, $\mu _{0}$.

In view of (\ref{SB45}) and (\ref{SB47}), Scheme (\ref{SB20}) will be stable
if the following scheme be stable (see \cite[pp. 390-392]{Samarskii 2001}, 
\cite[Theorem 7]{Borisov and Mond 2008}). 
\begin{equation*}
\delta \mathbf{v}_{i+0.5}^{n+0.5}=0.5\left( \delta \mathbf{v}_{i}^{n}+\delta 
\mathbf{v}_{i+1}^{n}\right) +\frac{\varkappa }{8}\mathbb{D}_{i}^{n}\cdot
\left( \delta \mathbf{v}_{i}^{n}-\delta \mathbf{v}_{i+1}^{n}\right) +
\end{equation*}
\begin{equation*}
\xi \frac{\left( \Delta t\right) ^{2}}{8\Delta x^{2}}\left[ \left( \mathbf{A}%
_{i+1}^{n}\right) ^{2}\cdot \mathbb{B}_{i}-\left( \mathbf{A}_{i}^{n}\right)
^{2}\cdot \mathbb{A}_{i}\right] \cdot \left( \delta \mathbf{v}%
_{i+1}^{n}-\delta \mathbf{v}_{i}^{n}\right) +
\end{equation*}
\begin{equation}
\frac{\Delta t}{2\Delta x}\left( \mathbf{A}_{i}^{n}\cdot \delta \mathbf{v}%
_{i}^{n}-\mathbf{A}_{i+1}^{n}\cdot \delta \mathbf{v}_{i+1}^{n}\right) .
\label{SB50}
\end{equation}
We rewrite (\ref{SB50}) to read 
\begin{equation}
\delta \mathbf{v}_{i+0.5}^{n+0.5}=0.5\left( \mathbf{I+E}_{i}^{n}\right)
\cdot \delta \mathbf{v}_{i}^{n}+0.5\left( \mathbf{I-E}_{i+1}^{n}\right)
\cdot \delta \mathbf{v}_{i+1}^{n},  \label{SB60}
\end{equation}
where 
\begin{equation}
\mathbf{E}_{i}^{n}=\frac{\varkappa }{4}\mathbb{D}_{i}^{n}-\xi \frac{\left(
\Delta t\right) ^{2}}{4\left( \Delta x\right) ^{2}}\left[ \left( \mathbf{A}%
_{i+1}^{n}\right) ^{2}\cdot \mathbb{B}_{i}-\left( \mathbf{A}_{i}^{n}\right)
^{2}\cdot \mathbb{A}_{i}\right] +\frac{\Delta t}{\Delta x}\mathbf{A}_{i}^{n},
\label{SB70}
\end{equation}
\begin{equation}
\mathbf{E}_{i+1}^{n}=\frac{\varkappa }{4}\mathbb{D}_{i}^{n}-\xi \frac{\left(
\Delta t\right) ^{2}}{4\left( \Delta x\right) ^{2}}\left[ \left( \mathbf{A}%
_{i+1}^{n}\right) ^{2}\cdot \mathbb{B}_{i}-\left( \mathbf{A}_{i}^{n}\right)
^{2}\cdot \mathbb{A}_{i}\right] +\frac{\Delta t}{\Delta x}\mathbf{A}%
_{i+1}^{n}.  \label{SB80}
\end{equation}
Since the operator $\mathbf{A}$ $(=\partial \mathbf{f}\left( \mathbf{u}%
\right) \diagup \partial \mathbf{u)}$ is Fr\'{e}chet-differentiable, and its
derivative is bounded, see (\ref{SB05}), we get, by virtue of (\ref{SB40})
and \cite[Corollary 3.2.4]{Ortega and Rheinboldt 1970}, that 
\begin{equation}
\left\| \mathbf{E}_{i+1}^{n}-\mathbf{E}_{i}^{n}\right\| =\frac{\Delta t}{%
\Delta x}\left\| \mathbf{A}_{i+1}^{n}-\mathbf{A}_{i}^{n}\right\| \leq \frac{%
\Delta t}{\Delta x}\alpha _{A}\left\| \delta \mathbf{v}_{i}^{n}\right\| \leq
\alpha _{A}C_{v}\Delta t.  \label{SB85}
\end{equation}
We find, in view of the first inequality in (\ref{SB05}) and (\ref{CA203}),
that the spectrum $s\left( \mathbf{E}_{i}^{n}\right) \subset \left[ -\lambda
_{E},\lambda _{E}\right] $, where 
\begin{equation}
\lambda _{E}=\varkappa \aleph -\xi \left( \frac{\Delta t\lambda _{\max }}{%
\Delta x}\right) ^{2}\aleph +\frac{\Delta t}{\Delta x}\lambda _{\max },\quad
\forall i,n.  \label{SB90}
\end{equation}
Hence, by virtue of \cite[Theorem 7]{Borisov and Mond 2008} we find that the
scheme (\ref{SB60}) will be stable if 
\begin{equation}
\underset{\lambda \in \left[ -\lambda _{E},\lambda _{E}\right] }{\max }%
0.5\left( \left| {1+\lambda }\right| +\left| {1-\lambda }\right| \right)
\leqslant 1,\quad \forall i,n.  \label{SB95}
\end{equation}
We obtain from (\ref{SB95}) the following condition for the stability of the
variational scheme (\ref{SB20}) 
\begin{equation}
\left( \varkappa -\xi C_{r}^{2}\right) \aleph +C_{r}\leq 1,\quad C_{r}=\frac{%
\Delta t\lambda _{\max }}{\Delta x}\leq 1.  \label{SB100}
\end{equation}
Thus, in view of Theorem \ref{iffStab} (see also \cite[Theorem 3]{Borisov
and Mond 2008}), the scheme (\ref{SA30}) will be stable if (\ref{SB100}) be
valid.

Let us note that the parameters $\varkappa $ and $\xi $ are taken as
constant in Scheme (\ref{SA30}). However, in practice, it can be convenient
to take that $\varkappa _{i}^{n}=\varkappa (\mathbf{v}_{i}^{n})$ and $\xi
_{i}^{n}=\xi (\mathbf{v}_{i}^{n})$. In such a case the condition, (\ref
{SB100}), for the stability of (\ref{SA30}) can be grounded in perfect
analogy to the above, if 
\begin{equation}
\left\| \delta \varkappa _{i}^{n}\right\| =\left\| \frac{\partial \varkappa
_{i}^{n}}{\partial \mathbf{v}_{i}^{n}}\cdot \delta \mathbf{v}%
_{i}^{n}\right\| \leq \alpha _{\varkappa }\left\| \delta \mathbf{v}%
_{i}^{n}\right\| ,\ \left\| \delta \xi _{i}^{n}\right\| =\left\| \frac{%
\partial \xi _{i}^{n}}{\partial \mathbf{v}_{i}^{n}}\cdot \delta \mathbf{v}%
_{i}^{n}\right\| \leq \alpha _{\xi }\left\| \delta \mathbf{v}%
_{i}^{n}\right\| ,  \label{SB110}
\end{equation}
where $\alpha _{\varkappa }$, $\alpha _{\xi }$ $=$ $const$.

\subsection{Operator splitting schemes\label{OSS}}

By virtue of the operator-splitting idea \cite{Bereux and Sainsaulieu 1997}, 
\cite{Du Tao et al. 2003}, \cite{Gosse L. 2000}, \cite{LeVeque 2002} (see
also LOS in \cite{Samarskii 2001}), the following chain of equations
corresponds to the problem (\ref{INA10}) 
\begin{equation}
\frac{1}{2}\frac{\partial \mathbf{U}}{\partial t}=\frac{1}{\tau }\mathbf{q}%
\left( \mathbf{U}\right) ,\quad t_{n}<t\leq t_{n+0.5},\quad \mathbf{U}\left( 
\mathbf{x},t_{n}\right) =\mathbf{U}^{n}\left( \mathbf{x}\right) ,
\label{OS20}
\end{equation}
\begin{equation}
\frac{1}{2}\frac{\partial \mathbf{U}}{\partial t}+\sum_{j=1}^{N}\frac{%
\partial }{\partial x_{j}}\mathbf{f}_{j}\left( \mathbf{U}\right)
=0,\,t_{n+0.5}<t\leq t_{n+1},\,\mathbf{U}\left( \mathbf{x},t_{n+0.5}\right) =%
\mathbf{U}^{n+0.5}\left( \mathbf{x}\right) ,  \label{OS10}
\end{equation}
where $\mathbf{U}^{n}\left( \mathbf{x}\right) $ denotes the solution to (\ref
{OS10}) at $t=t_{n}$, $\mathbf{U}^{n+0.5}\left( \mathbf{x}\right) $ denotes
the solution to (\ref{OS20}) at $t=t_{n+0.5}$. If a high-resolution method
is used directly for the homogeneous conservation law (\ref{OS10}), then it
is natural to use a high-order scheme for (\ref{OS20}). As applied to, in
general, stiff ($\tau \ll 1$) System (\ref{INA10}), the second order schemes
can be constructed on the basis of operator-splitting techniques with ease
if (\ref{OS20}) will be approximated by an implicit scheme and (\ref{OS10})
by an explicit one, see Proposition \ref{First to Second order scheme} in
Section \ref{Appendix1}. As an example, let us develop a central scheme for
a 1-D version of (\ref{INA10}). After operator-splitting, the 1-D equation
can be represented in the form 
\begin{equation}
\frac{1}{2}\frac{\partial \mathbf{U}}{\partial t}=\frac{1}{\tau }\mathbf{q}%
\left( \mathbf{U}\right) ,\quad t_{n}<t\leq t_{n+0.25},\quad \mathbf{U}%
\left( x,t_{n}\right) =\mathbf{U}^{n}\left( x\right) ,  \label{Os22}
\end{equation}
\begin{equation}
\frac{1}{2}\frac{\partial \mathbf{U}}{\partial t}+\frac{\partial }{\partial x%
}\mathbf{f}\left( \mathbf{U}\right) =0,\ t_{n+0.25}<t\leq t_{n+0.5},\ 
\mathbf{U}\left( x,t_{n+0.25}\right) =\mathbf{U}^{n+0.25}\left( x\right) .
\label{OS40}
\end{equation}
Let us first consider the case when the following first-order implicit
scheme be used for (\ref{Os22}) 
\begin{equation}
\mathbf{v}_{i}^{n+0.25}=\mathbf{v}_{i}^{n}+\frac{\Delta t}{2\tau }\mathbf{q}%
\left( \mathbf{v}_{i}^{n+0.25}\right) ,  \label{OS30}
\end{equation}
and a central scheme with nonstaggered grid cells will be used for (\ref
{OS40}). To eliminate the staggering in (\ref{CA210}), we can define, e.g. 
\cite{Jiang et al. 1998}, the nonstaggered cell-average as the average of
its two neighboring staggered cell-averages. Then, by virtue of (\ref{CA210}%
), we find 
\begin{equation*}
\mathbf{v}_{i}^{n+0.5}=0.25\left( \mathbf{v}_{i-1}^{n+0.25}+2\mathbf{v}%
_{i}^{n+0.25}+\mathbf{v}_{i+1}^{n+0.25}\right) -\varkappa \frac{\Delta x}{16}%
\left( \mathbf{d}_{i+1}^{n+0.25}-\mathbf{d}_{i-1}^{n+0.25}\right) -
\end{equation*}
\begin{equation}
\frac{\Delta t}{4\Delta x}\left( \mathbf{f}_{i+1}^{n+0.25}-\mathbf{f}%
_{i-1}^{n+0.25}\right) .  \label{OS60}
\end{equation}
It is clear that Scheme (\ref{OS60}) approximates (\ref{OS40}) with the
accuracy $O(\Delta t+\left( \Delta x\right) ^{2})$, however, in view of
Proposition \ref{First to Second order scheme} in Section \ref{Appendix1},
Scheme (\ref{OS30})-(\ref{OS60}), taken as a whole, is of the second order
approximation for the 1-D version of (\ref{INA10}).

Let us develop another nonstaggered central scheme approximating a 1-D
version of (\ref{INA10}) with the accuracy $O(\left( \Delta t\right)
^{2}+\left( \Delta x\right) ^{2})$ and such that its components (after
operator splitting) will be of the second order. It can be done on the basis
of the second order scheme (\ref{OS30}), (\ref{OS60}) with ease. Actually,
adding to and subtracting from Equation (\ref{FS15}) (see Section \ref
{Appendix1}, Proposition \ref{First to Second order scheme}), rewritten for $%
t_{n}<t\leq t_{n+0.5}$, the same quantity, we obtain (after operator
splitting) the following scheme, instead of (\ref{OS30}), (\ref{OS60}), 
\begin{equation}
\mathbf{v}_{i}^{n+0.25}=\mathbf{v}_{i}^{n}+\frac{\Delta t}{\tau }\mathbf{q}%
_{i}^{n+0.25}-\frac{\left( \Delta t\right) ^{2}}{32}\left( \frac{\partial
^{2}\mathbf{U}}{\partial t^{2}}\right) _{i}^{n+0.25},  \label{SOS10}
\end{equation}
\begin{equation*}
\mathbf{v}_{i}^{n+0.5}=0.25\left( \mathbf{v}_{i-1}^{n+0.25}+2\mathbf{v}%
_{i}^{n+0.25}+\mathbf{v}_{i+1}^{n+0.25}\right) -\varkappa \frac{\Delta x}{16}%
\left( \mathbf{d}_{i+1}^{n+0.25}-\mathbf{d}_{i-1}^{n+0.25}\right) +
\end{equation*}
\begin{equation}
\frac{\left( \Delta t\right) ^{2}}{32}\left( \frac{\partial ^{2}\mathbf{U}}{%
\partial t^{2}}\right) _{i}^{n+0.25}-\frac{\Delta t}{4\Delta x}\left( 
\mathbf{f}_{i+1}^{n+0.25}-\mathbf{f}_{i-1}^{n+0.25}\right) .  \label{SOS20}
\end{equation}
Thus, Scheme (\ref{SOS10}) as well as Scheme (\ref{SOS20}) are of the second
order, and Scheme (\ref{SOS10})-(\ref{SOS20}), taken as a whole, is of the
second order as well.

Using Taylor series expansion, and central differencing, we find 
\begin{equation*}
\mathbf{v}_{i}^{n+0.125}=\mathbf{v}_{i}^{n+0.25}-\frac{\Delta t}{8}\left( 
\frac{\partial \mathbf{U}}{\partial t}\right) _{i}^{n+0.25}+
\end{equation*}
\begin{equation}
\frac{1}{2}\left( \frac{\Delta t}{8}\right) ^{2}\left( \frac{\partial ^{2}%
\mathbf{U}}{\partial t^{2}}\right) _{i}^{n+0.25}+O\left( \left( \Delta
t\right) ^{3}\right) ,  \label{SOS25}
\end{equation}
\begin{equation}
\mathbf{v}_{i}^{n+0.25}=\mathbf{v}_{i}^{n}+\frac{\Delta t}{4}\left( \frac{%
\partial \mathbf{U}}{\partial t}\right) _{i}^{n+0.125}+O\left( \left( \Delta
t\right) ^{3}\right) .  \label{SOS27}
\end{equation}
We obtain, by virtue of (\ref{SA20}), (\ref{OS40}), that 
\begin{equation}
\frac{\partial ^{2}\mathbf{U}}{\partial t^{2}}=-2\frac{\partial }{\partial t}%
\left( \frac{\partial \mathbf{f}}{\partial x}\right) =-2\frac{\partial }{%
\partial x}\left( \frac{\partial \mathbf{f}}{\partial t}\right) =4\frac{%
\partial }{\partial x}\left( \mathbf{A}^{2}\cdot \frac{\partial \mathbf{U}}{%
\partial x}\right) ,  \label{SOS40}
\end{equation}
where $\mathbf{A=}\partial \mathbf{f\diagup }\partial \mathbf{U}$. Then 
\begin{equation*}
\left[ \frac{\partial }{\partial x}\left( \mathbf{A}^{2}\cdot \frac{\partial 
\mathbf{U}}{\partial x}\right) \right] _{i}^{n+0.25}=\frac{1}{\Delta x}\left[
\left( \mathbf{A}_{i+0.5}^{n+0.25}\right) ^{2}\cdot \frac{\mathbf{v}%
_{i+1}^{n+0.25}-\mathbf{v}_{i}^{n+0.25}}{\Delta x}\right. -
\end{equation*}
\begin{equation}
\left. \left( \mathbf{A}_{i-0.5}^{n+0.25}\right) ^{2}\cdot \frac{\mathbf{v}%
_{i}^{n+0.25}-\mathbf{v}_{i-1}^{n+0.25}}{\Delta x}\right] +O\left( \left(
\Delta x\right) ^{2}\right) ,  \label{SOS50}
\end{equation}
where $(\mathbf{A}_{i+0.5}^{n+0.25})^{2}=0.5\left( (\mathbf{A}%
_{i+1}^{n+0.25})^{2}+(\mathbf{A}_{i}^{n+0.25})^{2}\right) $. By virtue of (%
\ref{OS20}), (\ref{SOS25})-(\ref{SOS50}), we rewrite Scheme (\ref{SOS10})-(%
\ref{SOS20}) to read 
\begin{equation}
\mathbf{v}_{i}^{n+0.125}=\mathbf{v}_{i}^{n+0.25}-\frac{\Delta t}{8\tau }%
\left( \mathbf{q}_{i}^{n+0.125}+\mathbf{q}_{i}^{n+0.25}\right) ,
\label{SOS55}
\end{equation}
\begin{equation}
\mathbf{v}_{i}^{n+0.25}=\mathbf{v}_{i}^{n}+\frac{\Delta t}{2\tau }\mathbf{q}%
_{i}^{n+0.125},  \label{SOS60}
\end{equation}
\begin{equation*}
\mathbf{v}_{i}^{n+0.5}=0.25\left( \mathbf{v}_{i-1}^{n+0.25}+2\mathbf{v}%
_{i}^{n+0.25}+\mathbf{v}_{i+1}^{n+0.25}\right) -\varkappa \frac{\Delta x}{16}%
\left( \mathbf{d}_{i+1}^{n+0.25}-\mathbf{d}_{i-1}^{n+0.25}\right) +
\end{equation*}
\begin{equation*}
\frac{\xi \left( \Delta t\right) ^{2}}{8\left( \Delta x\right) ^{2}}\left[
\left( \mathbf{A}_{i+0.5}^{n+0.25}\right) ^{2}\cdot \left( \mathbf{v}%
_{i+1}^{n+0.25}-\mathbf{v}_{i}^{n+0.25}\right) \right. -\left. \left( 
\mathbf{A}_{i-0.5}^{n+0.25}\right) ^{2}\cdot \left( \mathbf{v}_{i}^{n+0.25}-%
\mathbf{v}_{i-1}^{n+0.25}\right) \right]
\end{equation*}
\begin{equation}
-\frac{\Delta t}{4\Delta x}\left( \mathbf{f}_{i+1}^{n+0.25}-\mathbf{f}%
_{i-1}^{n+0.25}\right) ,  \label{SOS70}
\end{equation}
where $\xi $ is introduced in the third term in the right-hand side of (\ref
{SOS70}) by analogy with Scheme (\ref{SA30}), and, hence, $0\leq \xi \leq 1$%
. If $\xi =0$, then (\ref{SOS70}) coincides with (\ref{OS60}), being $%
O(\Delta t+\left( \Delta x\right) ^{2})$ accurate. If $\varkappa =1$ and $%
\xi =1$, then Scheme (\ref{SOS55})-(\ref{SOS70}) approximates the 1-D
version of (\ref{INA10}) with the accuracy $O(\left( \Delta t\right)
^{2}+\left( \Delta x\right) ^{2})$.

By analogy with the scheme COS2, (\ref{SA30}), we find the conditions for
the stability of Scheme (\ref{SOS70}) using its scheme in variations. Scheme
(\ref{SOS70}) will be stable if 
\begin{equation}
\varkappa \aleph -0.5\xi C_{r}^{2}+C_{r}\leq 1,\quad C_{r}=\frac{\Delta
t\lambda _{\max }}{\Delta x}\leq 1.  \label{SOS80}
\end{equation}

Let us note that in practice (e.g., \cite{LeVeque 2002}, \cite{Samarskii
2001}) the operator-splitting techniques find a wide range of application in
designing economical schemes for Eq. (\ref{OS10}) in the domains of
complicated geometry. The resulting method, in general, will be only
first-order accurate in time because of the splitting \cite{LeVeque 2002}, 
\cite{Samarskii 2001}. Thus, in line with established practice we will
replace the multidimensional Eq. (\ref{OS10}) by the chain of the
one-dimensional equations: 
\begin{equation}
\frac{1}{2N}\frac{\partial \mathbf{U}_{j}}{\partial t}+\frac{\partial }{%
\partial x_{j}}\mathbf{f}_{j}\left( \mathbf{U}_{j}\right) =0,\
t_{n+0.5+(j-1)\diagup (2N)}<t\leq t_{n+0.5+j\diagup (2N)},  \label{SOS90}
\end{equation}
where$\ \mathbf{U}_{j}\left( x,t_{n+0.5+(j-1)\diagup (2N)}\right) =\mathbf{U}%
_{j-1}\left( x,t_{n+0.5+(j-1)\diagup (2N)}\right) $, $\ j=1,2,\ldots N$, $%
\mathbf{U}_{0}\left( x,t_{n+0.5}\right) $ denotes the solution to (\ref{OS20}%
) at $t=t_{n+0.5}$. Eq. (\ref{OS20}) will be approximated by a first-order
implicit scheme or a second-order implicit Runge-Kutta scheme. In
particular, it will be used the following Runge-Kutta scheme 
\begin{equation}
\mathbf{v}_{i}^{n+0.25}=\mathbf{v}_{i}^{n+0.5}-\frac{\Delta t}{2\tau }%
\mathbf{q}\left( \mathbf{v}_{i}^{n+0.5}\right) ,\ \mathbf{v}_{i}^{n+0.5}=%
\mathbf{v}_{i}^{n}+\frac{\Delta t}{\tau }\mathbf{q}\left( \mathbf{v}%
_{i}^{n+0.25}\right) ,  \label{SOS100}
\end{equation}
since this scheme possesses a discrete analogy to the continuous asymptotic
limit.

\section{Exemplification and discussion\label{Exemplification and discussion}%
}

In this section, we are mainly concerned with verification of the second
order central scheme COS2, (\ref{SA30}).

\subsection{Scalar non-linear equation}

As the first stage in the verification, we will focus on the following
scalar 1-D version of the problem (\ref{INA10}): 
\begin{equation}
\frac{\partial u}{\partial t}+\frac{\partial }{\partial x}f\left( u\right)
=0,\quad x\in \mathbb{R},\ 0<t\leq T_{\max };\quad \left. u\left( x,t\right)
\right| _{t=0}=u^{0}\left( x\right) .  \label{VS10}
\end{equation}
We will solve the inviscid Burgers equation (i.e. $f\left( u\right) \equiv
u^{2}\diagup 2$) with the following initial condition 
\begin{equation}
u\left( x,0\right) =\left\{ 
\begin{array}{cc}
u_{0}, & x\in \left( h_{L},h_{R}\right) \\ 
0, & x\notin \left( h_{L},h_{R}\right)
\end{array}
\right. ,\quad h_{R}>h_{L},\ u_{0}=const\neq 0.  \label{VS70}
\end{equation}
The exact solution to (\ref{VS10}), (\ref{VS70}) is given by 
\begin{equation}
u\left( x,t\right) =\left\{ 
\begin{array}{cc}
u_{1}\left( x,t\right) , & 0<t\leq T \\ 
u_{2}\left( x,t\right) , & t>T
\end{array}
\right. ,  \label{VS80}
\end{equation}
where $T=2S\diagup u_{0}$, $S=h_{R}-h_{L}$, 
\begin{equation}
u_{1}\left( x,t\right) =\left\{ 
\begin{array}{cc}
\frac{x-h_{L}}{b-h_{L}}u_{0}, & h_{L}<x\leq b,\ b=u_{0}t+h_{L} \\ 
u_{0}, & b<x\leq 0.5u_{0}t+h_{R} \\ 
0, & x\leq h_{L}\ or\ x>0.5u_{0}t+h_{R}
\end{array}
\right. ,  \label{VS90}
\end{equation}
\begin{equation}
u_{2}\left( x,t\right) =\left\{ 
\begin{array}{cc}
\frac{2S\left( x-h_{L}\right) }{\left( L-h_{L}\right) ^{2}}u_{0}, & 
h_{L}<x\leq L \\ 
0, & x\leq h_{L}\ or\ x>L
\end{array}
\right. ,  \label{VS100}
\end{equation}
\begin{equation}
L=2\sqrt{S^{2}+0.5u_{0}S\left( t-T\right) }+h_{L}.  \label{VS110}
\end{equation}

First, it will be used Scheme (\ref{SA30}) under $\xi =0$, i.e. the first
order in time central scheme COS1, (\ref{CA210}). The numerical solutions
were computed on a uniform grid with spatial increments of $\Delta x=0.01$,
the velocity $u_{0}=1$ in (\ref{VS70}), $h_{L}=0.2$, $h_{R}=1$, the
monotonicity parameter $\aleph =0.5$, the CFL number $Cr$ $\equiv $ $%
u_{0}\Delta t\diagup \Delta x=$ $0.5$, and the parameter $\varkappa =1$ in (%
\ref{CA210}). The results of simulation are depicted with the exact solution
in Figure \ref{K2COS01}. 
\begin{figure}[h]
\centerline{\includegraphics[width=11.50cm,height=3.8cm]{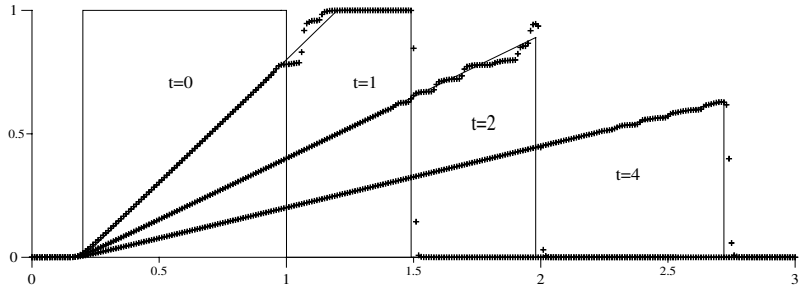}}
\caption{Inviscid Burgers equation. The scheme COS1 ($\varkappa =1$) versus
the analytical solution. Crosses: numerical solution; Solid line: analytical
solution and initial data. $C_{r}=\aleph =0.5$, $\Delta x=0.01$. }
\label{K2COS01}
\end{figure}
We note (Figure \ref{K2COS01}) that the first order scheme, (\ref{CA210}),
exhibits a typical second-order nature, however spurious solutions are
produced by the scheme. Notice, the numerical simulations were performed
with such values of the parameter $\varkappa $, CFL number, $Cr$, and
monotonicity parameter, $\aleph $, that (\ref{SB100}) was not violated. As
it can be seen in Figure \ref{K2COS01}, the boundary maximum principle is
not violated by the scheme, i.e., the maximum positive values of the
dependent variable, $v$, occur at the boundary $t=0$. It is interesting that
the spurious solution (see Figure \ref{K2COS01}) produced by the scheme COS1
has the monotonicity property \cite{Harten 1983}, since no new local extrema
in $x$ are created as well as the value of a local minimum is non-decreasing
and the value of a local maximum is non-increasing.

Let us note that the problem of building free-of-spurious-oscillations
schemes is, in general, unsettled up to the present. Even the best modern
high-resolution schemes can produce spurious oscillations, and these
oscillations are often of ENO type (see, e.g., \cite{Pareschi et al. 2005}
and references therein). We found that the oscillations produced by the COS1
scheme, (\ref{CA210}), are of ENO type, namely their amplitude decreases
rapidly with decreasing the time-increment $\Delta t$, and the oscillations
virtually disappear under a relatively low CFL number, $Cr$ $\leq $ $0.15$.
However, the reduction of the CFL number causes some smearing of the
solution. The spurious oscillations (see Figure \ref{K2COS01}) can be
eradicated without reduction CFL number, but decreasing the parameter $%
\varkappa $. Particularly, the spurious oscillations disappear if $\varkappa
=2\diagup 3$, $C_{r}=0.5$, however, this introduces more numerical smearing
than in the case of the CFL number reduction. Satisfactory results are
obtained under $\varkappa =0.82$ ($C_{r}=0.5$). The results of simulations
are not depicted here.

To gain insight to why the scheme COS1, (\ref{CA210}), can exhibit spurious
solutions, let us consider the, so called, \emph{first differential
approximation} of this scheme (\cite[p. 45]{Ganzha and Vorozhtsov 1996}, 
\cite[p. 376]{Samarskiy and Gulin 1973}; see also `modified equations' in 
\cite[p. 45]{Ganzha and Vorozhtsov 1996}, \cite{LeVeque 2002}, \cite{Morton
1996}). As reported in \cite{Ganzha and Vorozhtsov 1996}, \cite{Samarskiy
and Gulin 1973}, this heuristic method was originally presented by Hirt
(1968) (see \cite[p. 45]{Ganzha and Vorozhtsov 1996}) as well as by Shokin
and Yanenko (1968) (see \cite[p. 376]{Samarskiy and Gulin 1973}), and has
since been widely employed in the development of stable difference schemes
for PDEs.

We found that the local truncation error, $\psi $, for the scheme COS1 can
be written in the following form 
\begin{equation*}
\psi =\frac{\left( 1-\varkappa \right) \left( \Delta x\right) ^{2}}{4\Delta t%
}\frac{\partial ^{2}u\left( x,t\right) }{\partial x^{2}}+\frac{\Delta t}{4}%
\frac{\partial ^{2}f\left( u\right) }{\partial t\partial x}+
\end{equation*}
\begin{equation}
O\left( \frac{\left( \Delta x\right) ^{4}}{\Delta t}+\left( \Delta t\right)
^{2}+\left( \Delta x\right) ^{2}\right) .  \label{VS190}
\end{equation}
By virtue of (\ref{VS190}), we find the first differential approximation of
the scheme COS1 
\begin{equation}
\frac{\partial u}{\partial t}+\frac{\partial f\left( u\right) }{\partial x}=%
\frac{\Delta t}{4}\frac{\partial }{\partial x}\left( B\frac{\partial u\left(
x,t\right) }{\partial x}\right) ,  \label{VS200}
\end{equation}
where $B=\left( 1-\varkappa \right) \left( \Delta x\diagup \Delta t\right)
^{2}-A^{2}$. The term in right-hand side of (\ref{VS200}) will be
dissipative if 
\begin{equation}
\left( 1-\varkappa \right) \left( \frac{\Delta x}{\Delta t}\right)
^{2}-A^{2}>0,\ \Longrightarrow \ C_{r}^{2}<1-\varkappa .  \label{VS210}
\end{equation}
Thus, the scheme COS1, (\ref{CA210}), is non-dissipative under $\varkappa =1$%
, and hence can produce spurious oscillations. Notice, if $\varkappa =0.82$,
then we obtain from (\ref{VS210}) that $C_{r}<0.42$. Nevertheless, as it is
reported above, satisfactory results can be obtained under $C_{r}=0.5$ as
well.

So then, the notion of first differential approximation has enabled us to
understand that the spurious solutions exhibited by the scheme COS1, (\ref
{CA210}), are mainly associated with the negative numerical viscosity
introduced to obtain the scheme of the second order in space, i.e. $O(\left(
\Delta x\right) ^{2}+\Delta t)$. Let us consider the scheme COS2, (\ref{SA30}%
), approximating (\ref{VS10}) with the accuracy $O(\left( \Delta x\right)
^{2}+\left( \Delta t\right) ^{2})$. Notice, the second order scheme COS2, (%
\ref{SA30}), is nothing more than the scheme COS1, (\ref{CA210}), with the
additional non-negative numerical viscosity. To test the scheme COS2, (\ref
{SA30}), the inviscid Burgers equation was solved under the initial
condition (\ref{VS70}). The numerical solutions were computed under the same
values of parameters as in the case of the scheme COS1, but $C_{r}=1$. The
results of simulation are depicted with the exact solution in Figure \ref
{K2COS02}.

\begin{figure}[h]
\centerline{\includegraphics[width=11.50cm,height=3.8cm]{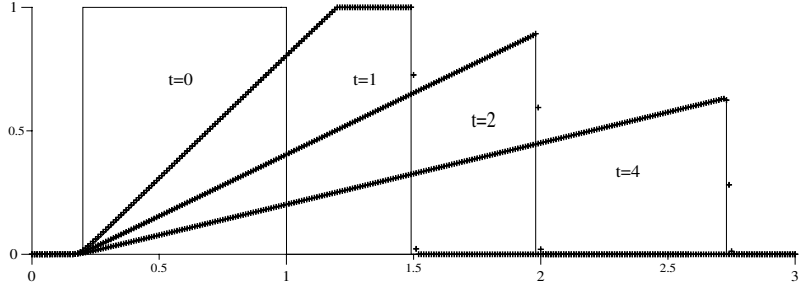}}
\caption{Inviscid Burgers equation. The scheme COS2 ($\protect\xi =1$, $%
\varkappa =1$) versus the analytical solution. Crosses: numerical solution;
Solid line: analytical solution and initial data. $C_{r}=1$, $\aleph =0.5$, $%
\Delta x=0.01$. }
\label{K2COS02}
\end{figure}

We note (Figure \ref{K2COS02}) that the scheme COS2, (\ref{SA30}), exhibits
a typical second-order nature without any spurious oscillations. Increasing
the value of $\aleph $ (up to $\aleph =1$) leads to a minor improvement of
the numerical solutions, whereas decreasing the value of $Cr$ leads to a
mild smearing of the solutions. The results of simulations are not depicted
here.

\subsection{Hyperbolic conservation laws with relaxation}

Let us consider the model system of hyperbolic conservation laws with
relaxation developed in \cite{Pember 1993a}: 
\begin{equation}
\frac{\partial w}{\partial t}+\frac{\partial }{\partial x}\left( \frac{1}{2}%
u^{2}+aw\right) =0,  \label{VS240}
\end{equation}
\begin{equation}
\frac{\partial z}{\partial t}+\frac{\partial }{\partial x}az=\frac{1}{\tau }%
Q(w,z),  \label{VS250}
\end{equation}
where 
\begin{equation}
Q(w,z)=z-m(u-u_{0}),\quad u=w-q_{0}z,  \label{VS260}
\end{equation}
$\tau $ denotes the relaxation time of the system, $q_{0}$, $m$, $a$, and $%
u_{0}$ are constants. The Jacobian, $\mathbf{A}$, can be written in the form 
\begin{equation}
\mathbf{A=}\left\{ 
\begin{array}{cc}
w-q_{0}z+a & -q_{0}\left( w-q_{0}z\right) \\ 
0 & a
\end{array}
\right\} .  \label{VS270}
\end{equation}
The system (\ref{VS240})-(\ref{VS250}) has the following frozen \cite{Pember
1993a} characteristic speeds $\lambda _{1}$ $=$ $a$, $\lambda _{2}$ $=$ $u+a$%
. The equilibrium equation for (\ref{VS240})-(\ref{VS250}) is 
\begin{equation}
\frac{\partial w}{\partial t}+\frac{\partial }{\partial x}\left( \frac{1}{2}%
u_{\ast }^{2}+aw\right) =0,  \label{VS280}
\end{equation}
where 
\begin{equation}
u_{\ast }=w-q_{0}z_{\ast },\quad z_{\ast }=\frac{m}{1+mq_{0}}\left(
w-u_{0}\right) .  \label{VS290}
\end{equation}
The equilibrium characteristic speed $\lambda _{\ast }$ can be written in
the form 
\begin{equation}
\lambda _{\ast }\left( w\right) =\frac{u_{\ast }\left( w\right) }{1+mq_{0}}%
+a.  \label{VS300}
\end{equation}

Pember's rarefaction test problem is to find the solution $\left\{
w,z\right\} $ to (\ref{VS240})-(\ref{VS250}), and hence the function $%
u=u\left( x,t\right) $, under $\tau \rightarrow 0$, and where 
\begin{equation}
\left\{ w,z\right\} =\left\{ 
\begin{array}{cc}
\left\{ w_{L},z_{\ast }\left( w_{L}\right) \right\} , & x<x_{0} \\ 
\left\{ w_{R},z_{\ast }\left( w_{R}\right) \right\} , & x>x_{0}
\end{array}
\right. ,  \label{VS310}
\end{equation}
\begin{equation}
0<u_{L}=w_{L}-q_{0}z_{\ast }\left( w_{L}\right) <u_{R}=w_{R}-q_{0}z_{\ast
}\left( w_{R}\right) .  \label{VS320}
\end{equation}
The analytical solution of this problem can be found in \cite{Pember 1993a}.
The parameters of the model system are assumed as follows: $q_{0}=-1$, $m=-1$%
, $u_{0}=3$, $a=\pm 1$, $\tau =10^{-8}$. The initial conditions of the
rarefaction problem are defined by 
\begin{equation}
u_{L}=2,\ \Longrightarrow \ z_{L}=m\left( u_{L}-u_{0}\right) =1,\
w_{L}=u_{L}+q_{0}z_{L}=1,  \label{VS330}
\end{equation}
\begin{equation}
u_{R}=3,\ \Longrightarrow \ z_{R}=m\left( u_{R}-u_{0}\right) =0,\
w_{R}=u_{R}+q_{0}z_{R}=3.  \label{VS340}
\end{equation}
The position of the initial discontinuity, $x_{0}$, is set according to the
value of $a$ so that the solutions of all the rarefaction problems are
identical \cite{Pember 1993a}. Let a position, $x_{R}^{t}$, of leading edge
or a position, $x_{L}^{t}$, of trailing edge of the rarefaction be known
(e.g., $x_{R}^{t}=0.85$, $x_{L}^{t}=0.7$ in \cite{Pember 1993a}), then 
\begin{equation}
x_{0}=x_{R}^{t}-\left( \frac{u_{R}}{1+mq_{0}}+a\right) t=x_{L}^{t}-\left( 
\frac{u_{L}}{1+mq_{0}}+a\right) t.  \label{VS350}
\end{equation}
At $t=0.3$, under (\ref{VS330})-(\ref{VS340}) we have \cite{Pember 1993a} 
\begin{equation}
u=\left\{ 
\begin{array}{cc}
2, & x\leq 0.7 \\ 
2+\frac{x-0.7}{0.85-0.7}, & 0.7<x<0.85 \\ 
3, & x\geq 0.85
\end{array}
\right. .  \label{VS360}
\end{equation}
The results of simulations, based upon the scheme COS2, (\ref{SA30}),
together with (\ref{SOS100}), under different values of the parameter $a$ ($%
a=1$, $a=-1$) and different values of a grid spacing, $\Delta x$, are
depicted in Figure \ref{RAREF1}.

\begin{figure}[tbph]
\centerline{\includegraphics[width=14.cm,height=6.cm]{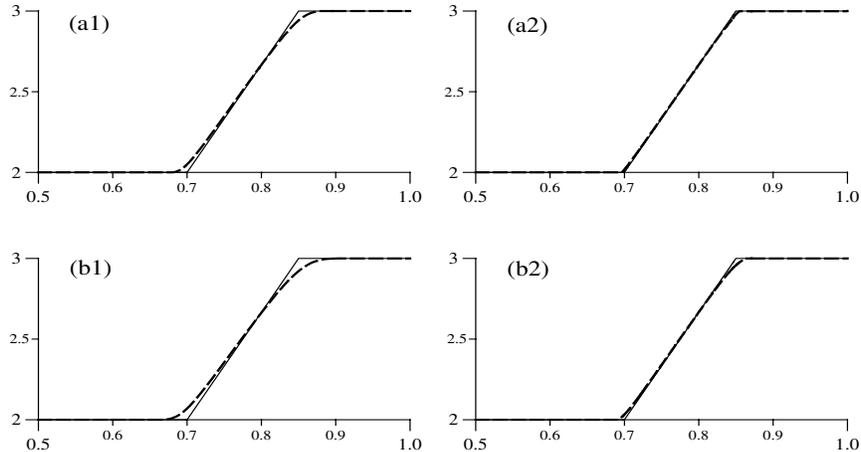}}
\caption{Pember's rarefaction test problem. The second-order scheme COS2 ($%
\protect\xi =1$, $\varkappa =1$) versus the analytical solution for $u$.
Dashed line: numerical solution; Solid line: analytical solution. Time $%
t=0.3 $, Courant number $C_{r}=1$, monotonicity parameter $\aleph =1$. (a1): 
$\Delta x=10^{-3}$, $a=1$; (a2): $\Delta x=2.5\times 10^{-4}$, $a=1$; (b1): $%
\Delta x=10^{-3}$, $a=-1$; (b2): $\Delta x=2.5\times 10^{-4}$, $a=-1$.}
\label{RAREF1}
\end{figure}

One can clearly see (Figure \ref{RAREF1}) that the scheme COS2 is free from
spurious oscillations. Let us also note that the results generated by the
scheme COS2 are less accurate in the case of negative value of $a$ than
those in the case of positive value of $a$. Specifically, in the numerical
solutions produced under $a=-1$, the representations of the trailing and
leading edges of the rarefaction are more smeared than those in the
solutions produced under $a=1$. Notice, under some negative value of $a$,
the frozen and the equilibrium characteristic speeds do not all have the
same sign.

\subsection{1-D Euler equation of gas dynamics}

In this subsection we apply the second order scheme COS2, (\ref{SA30}), to
the Euler equations of gamma-law gas: 
\begin{equation}
\frac{\partial \mathbf{u}\left( x,t\right) }{\partial t}+\frac{\partial }{%
\partial x}\mathbf{F}\left( \mathbf{u}\right) =0,\quad x\in \mathbb{R},\
t>0;\quad \mathbf{u}\left( x,0\right) =\mathbf{u}^{0}\left( x\right) ,
\label{VE10}
\end{equation}
\begin{equation}
\mathbf{u\equiv }\left\{ u_{1},u_{2},u_{3}\right\} ^{T}=\left\{ \rho ,\rho
v,e\right\} ^{T},\quad \mathbf{F}\left( \mathbf{u}\right) =\left\{ \rho
v,\rho v^{2}+p,\left( e+p\right) v\right\} ^{T},  \label{VE20}
\end{equation}
\begin{equation}
e=\frac{p}{\gamma -1}+\frac{1}{2}\rho v^{2},\quad \gamma =const,
\label{VE30}
\end{equation}
where $\rho $, $v$, $p$, $e$ denote the density, velocity, pressure, and
total energy respectively. We consider the Riemann problem subject to
Riemann initial data 
\begin{equation}
\mathbf{u}^{0}\left( x\right) =\left\{ 
\begin{array}{cc}
\mathbf{u}_{L} & x<x_{0} \\ 
\mathbf{u}_{R} & x>x_{0}
\end{array}
\right. ,\quad \mathbf{u}_{L},\mathbf{u}_{R}=const.  \label{VE40}
\end{equation}
The analytic solution to the Riemann problem can be found in \cite[Sec. 14]
{LeVeque 2002}.

First we solve the shock tube problem (see, e.g., \cite{Balaguer and Conde
2005}, \cite{LeVeque 2002}, \cite{Liu and Tadmor 1998}) with Sod's initial
data: 
\begin{equation}
\mathbf{u}_{L}=\left\{ 
\begin{array}{c}
1 \\ 
0 \\ 
2.5
\end{array}
\right\} ,\quad \mathbf{u}_{R}=\left\{ 
\begin{array}{c}
0.125 \\ 
0 \\ 
0.25
\end{array}
\right\} .  \label{VE50}
\end{equation}
Following Balaguer and Conde \cite{Balaguer and Conde 2005} as well as Liu
and Tadmor \cite{Liu and Tadmor 1998} we assume that the computational
domain is $0\leq x\leq 1$; the point $x_{0}$ is located at the middle of the
interval $\left[ 0,1\right] $, i.e. $x_{0}=0.5$; the equations (\ref{VE10})
are integrated up to $t=0.16$ on a spatial grid with 200 nodes as in \cite
{Balaguer and Conde 2005} and in \cite{Liu and Tadmor 1998}. The CFL number
is taken to be $Cr=1$ in contrast to \cite{Balaguer and Conde 2005} and \cite
{Liu and Tadmor 1998}, where the simulations were done under $\Delta
t=0.1\Delta x$ (i.e. $0.13\lesssim Cr\lesssim 0.22$). The results of
simulations are depicted in Figure \ref{C1A5K8S}. 
\begin{figure}[tbph]
\centerline{\includegraphics[width=11.9cm,height=15.4cm]{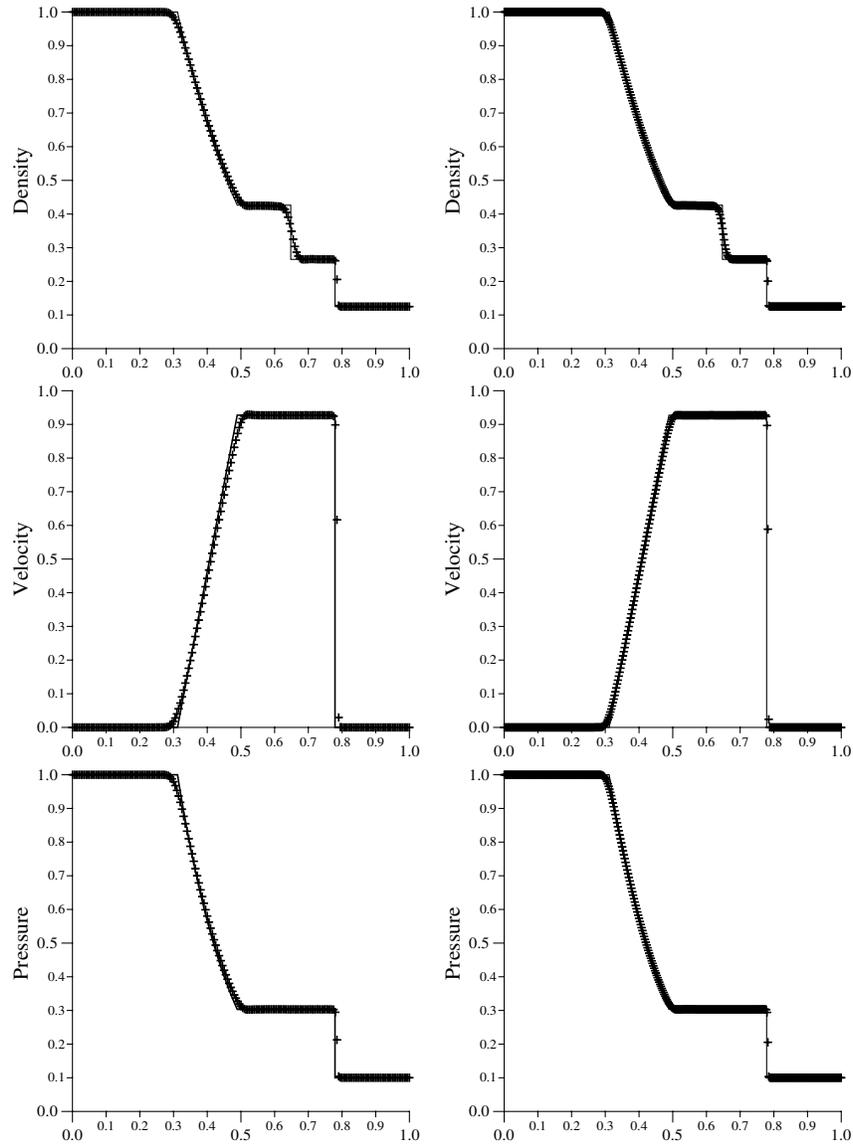}}
\caption{Sod's problem. The scheme COS2 under $C_{r}=1$, $\aleph =0.5$, $%
\protect\xi =1$, $\varkappa =0.8$ versus the analytical solution. Time $%
t=0.16$, spatial increment $\Delta x=0.005$ (left column) and $\Delta
x=0.0025$ (right column). }
\label{C1A5K8S}
\end{figure}

The results depicted in Figure \ref{C1A5K8S} (left column) are not worse in
comparison to the corresponding third-order central results of \cite[p. 418]
{Liu and Tadmor 1998} as well as to the results obtained by the fourth-order
non-oscillatory scheme in \cite[p. 472]{Balaguer and Conde 2005}. Notice,
the fourth-order scheme \cite[p. 472]{Balaguer and Conde 2005} gives a
better resolution but, in contrast to the scheme COS2, can produce spurious
oscillations.

Let us also note that the number of multiplications and divisions per one
grid node in Scheme COS2, (\ref{SA30}), and in the second-order scheme
considered in \cite{Liu and Tadmor 1998} is approximately the same, but less
than this number in the third- and fourth-order schemes developed in 
\cite[p. 418]{Liu and Tadmor 1998}, \cite[p. 472]{Balaguer and Conde 2005},
respectively. Hence, the results depicted in Figure \ref{C1A5K8S} (right
column) is rather chipper (in terms of CPU time) than the ones demonstrated
in \cite[p. 418]{Liu and Tadmor 1998}, \cite[p. 472]{Balaguer and Conde 2005}%
. Along with loss of computational efficiency, simulations with low CFL
number can, in general, lead to excessive numerical smearing. As it is
demonstrated above, Scheme COS2 is free from such drawbacks.

\subsection{3-D axial symmetric gas dynamics}

We consider an adiabatic expansion of a gas plume into vacuum \cite{Anisimov
et al. 1993}, i.e., the so called Anisimov's problem. Taking into account
the symmetry of the plume with respect to the axis $z$, the gas-dynamic
equations can be written (for $0$ $<\ r$, $z$ $<$ $\infty $) as follows. 
\begin{equation}
\frac{\partial \rho }{\partial t}+\frac{1}{r}\frac{\partial \left( r\rho
v_{r}\right) }{\partial r}+\frac{\partial \left( \rho v_{z}\right) }{%
\partial z}=0,  \label{FA10}
\end{equation}
\begin{equation}
\frac{\partial }{\partial t}\left( \rho v_{r}\right) +\frac{1}{r}\frac{%
\partial }{\partial r}\left[ r\rho \left( v_{r}\right) ^{2}\right] +\frac{%
\partial }{\partial z}\left( \rho v_{z}v_{r}\right) +\frac{\partial p}{%
\partial r}=0,  \label{FA20}
\end{equation}
\begin{equation}
\frac{\partial }{\partial t}\left( \rho v_{z}\right) +\frac{\partial }{%
\partial z}\left[ \rho \left( v_{z}\right) ^{2}\right] +\frac{1}{r}\frac{%
\partial }{\partial r}\left( r\rho v_{z}v_{r}\right) +\frac{\partial p}{%
\partial z}=0,  \label{FA30}
\end{equation}
\begin{equation}
\frac{\partial \rho E}{\partial t}+\frac{1}{r}\frac{\partial }{\partial r}%
\left[ rv_{r}\left( \rho E+p\right) \right] +\frac{\partial }{\partial z}%
\left[ v_{z}\left( \rho E+p\right) \right] =0.  \label{FA40}
\end{equation}
\begin{equation}
\rho E=\frac{P}{\gamma -1}+0.5\rho v^{2},\ v^{2}=v_{r}^{2}+v_{z}^{2},\
\gamma =const.  \label{FA50}
\end{equation}
The initial conditions are the following (in details, see \cite{Anisimov et
al. 1993}) 
\begin{equation}
\rho =\rho \left( r,z\right) ,\ p\diagup \rho ^{\gamma }=const,\
v_{r}=v_{z}=0,\ r,z\geq 0,\ t=0.  \label{F160}
\end{equation}
At $r=0$ we assume that the axis $z$ is a reflection line. It prohibits any
normal flux of mass through the boundary $r=0$, i.e. 
\begin{equation}
v_{r}=0,\quad r=0,\ z\geq 0.  \label{FA55}
\end{equation}
Moreover, it is assumed that the pressure ($p$), density ($\rho $), and
tangential velocity ($v_{z}$) are even functions of normal distance to the
axis $z$ while the normal velocity ($v_{r}$) is an odd function of $r$. It
is also assumed that the plane $z=0$ is a reflection surface, i.e. the
pressure ($p$), density ($\rho $), and tangential velocity ($v_{r}$) are
even functions of normal distance above the target surface while the normal
velocity ($v_{z}$) is an odd function of $z$. The analytic solution to the
problem (\ref{FA10})-(\ref{FA55}) can be found in \cite{Anisimov et al. 1993}%
.

Notice, every point on the axis $r=0$ is a singular point for System (\ref
{FA10})-(\ref{FA40}). Assuming that all terms at (\ref{FA10}) are bounded
values at a vicinity of $r=0$, we find that $v_{r}\rightarrow 0$ as $%
r\rightarrow 0$. Hence 
\begin{equation}
\underset{r\rightarrow 0+0}{\lim }\frac{\left. \rho v_{r}\right| _{r>0}}{r}=%
\underset{r\rightarrow 0+0}{\lim }\frac{\left. \rho v_{r}\right|
_{r>0}-\left. \rho v_{r}\right| _{r=0}}{r}=\frac{\partial \left( \rho
v_{r}\right) }{\partial r}.  \label{FA69-10}
\end{equation}
In perfect analogy we obtain 
\begin{equation}
\frac{\rho \left( v_{r}\right) ^{2}}{r}\rightarrow \frac{\partial \rho
\left( v_{r}\right) ^{2}}{\partial r},\ \frac{\rho v_{z}v_{r}}{r}\rightarrow 
\frac{\partial \rho v_{z}v_{r}}{\partial r},\ \frac{v_{r}\left( \rho
E+p\right) }{r}\rightarrow \frac{\partial v_{r}\left( \rho E+p\right) }{%
\partial r},  \label{FA69-20}
\end{equation}
as $r\rightarrow 0$. By virtue of (\ref{FA69-10})-(\ref{FA69-20}), we obtain
from (\ref{FA10})-(\ref{FA40}) the following conditions at $r=0$: 
\begin{equation}
\frac{\partial \rho }{\partial t}+\frac{\partial }{\partial r}\left( 2\rho
v_{r}\right) +\frac{\partial }{\partial z}\left( \rho v_{z}\right) =0,
\label{FA69-30}
\end{equation}
\begin{equation}
\frac{\partial }{\partial t}\left( \rho v_{r}\right) +\frac{\partial }{%
\partial r}\left[ 2\rho \left( v_{r}\right) ^{2}+p\right] +\frac{\partial }{%
\partial z}\left( \rho v_{z}v_{r}\right) =0,  \label{FA69-40}
\end{equation}
\begin{equation}
\frac{\partial }{\partial t}\left( \rho v_{z}\right) +\frac{\partial }{%
\partial z}\left[ \rho \left( v_{z}\right) ^{2}+p\right] +\frac{\partial }{%
\partial r}\left( 2\rho v_{z}v_{r}\right) =0,  \label{FA69-50}
\end{equation}
\begin{equation}
\frac{\partial \rho E}{\partial t}+\frac{\partial }{\partial r}\left[
2v_{r}\left( \rho E+p\right) \right] +\frac{\partial }{\partial z}\left[
v_{z}\left( \rho E+p\right) \right] =0.  \label{FA69-60}
\end{equation}

In the analytic solution \cite{Anisimov et al. 1993} of the problem (\ref
{FA10})-(\ref{FA55}) the following input data are required: the initial
dimensions of the plume, $R_{0}$ and $Z_{0}$, its mass $M_{P}$, and the
initial energy $E_{P}$. We will use the following values as the reference
quantities: $l_{\ast }$ $=$ $R_{0}$, $v_{\ast }$ $=$ $\sqrt{\left( 5\gamma
-3\right) E_{P}\diagup M_{P}}$, $t_{\ast }$ $=$ $l_{\ast }\diagup v_{\ast }$%
, $\rho _{\ast }$ $=$ $M_{P}\diagup \left( R_{0}^{2}Z_{0}\right) $, $p_{\ast
}$ $=$ $\rho _{\ast }v_{\ast }^{2}$.

The equations (\ref{FA10})-(\ref{FA40}) are integrated up to $t=0.4$ with $%
\sigma \equiv Z_{0}\diagup R_{0}=0.1$. The CFL number is taken to be $Cr=1$.
It is assumed that the spatial increments are the following: $\Delta
r=0.0025 $, $\Delta z=0.00025$ if $0<t\leq 0.05$; $\Delta r=0.0025$, $\Delta
z=0.0005$ if $0.05<t\leq 0.1$; $\Delta r=0.005$, $\Delta z=0.001$ if $%
0.1<t\leq 0.2$; $\Delta r=0.01$, $\Delta z=0.002$ if $0.2<t\leq 0.4$. The
results of simulations as well as the analytical solution are depicted in
Figures \ref{DP1T0504}, \ref{MZR1T0504}. 
\begin{figure}[tbph]
\centerline{\includegraphics[width=12.75cm,height=16.5cm]{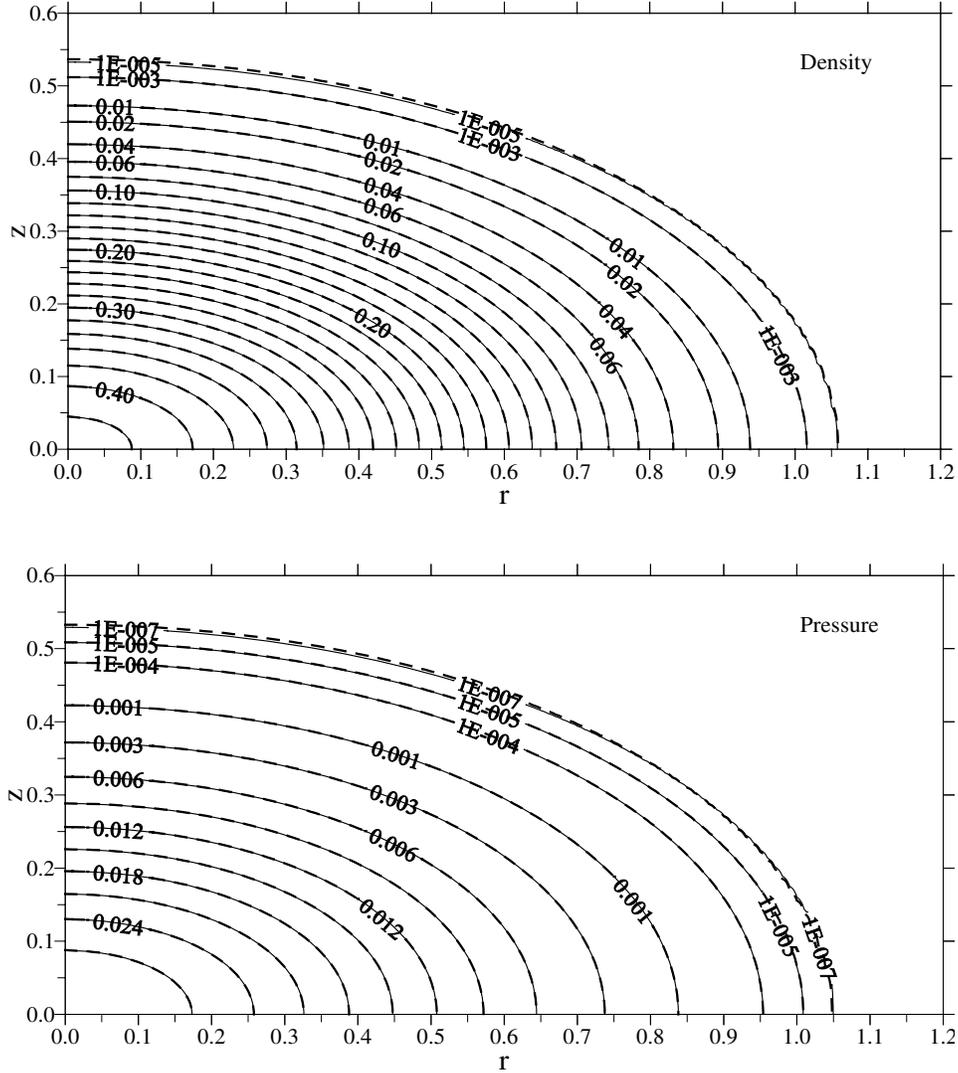}}
\caption{Anisimov's problem, density and pressure distribution. COS2 scheme
versus the analytical solution. $\protect\sigma \equiv Z_{0}\diagup
R_{0}=0.1 $, time $t=0.4$, CFL number $C_{r}=1$, monotonicity parameter $%
\aleph =\varkappa =1$, spatial increments: $\Delta r=0.0025$, $\Delta
z=0.00025$ if $0<t\leq 0.05$; $\Delta r=0.0025$, $\Delta z=0.0005$ if $%
0.05<t\leq 0.1$; $\Delta r=0.005$, $\Delta z=0.001$ if $0.1<t\leq 0.2$; $%
\Delta r=0.01$, $\Delta z=0.002$ if $0.2<t\leq 0.4$. Dashed lines: numerical
solution; Solid lines: analytical solution.}
\label{DP1T0504}
\end{figure}
\begin{figure}[tbph]
\centerline{\includegraphics[width=12.75cm,height=16.5cm]{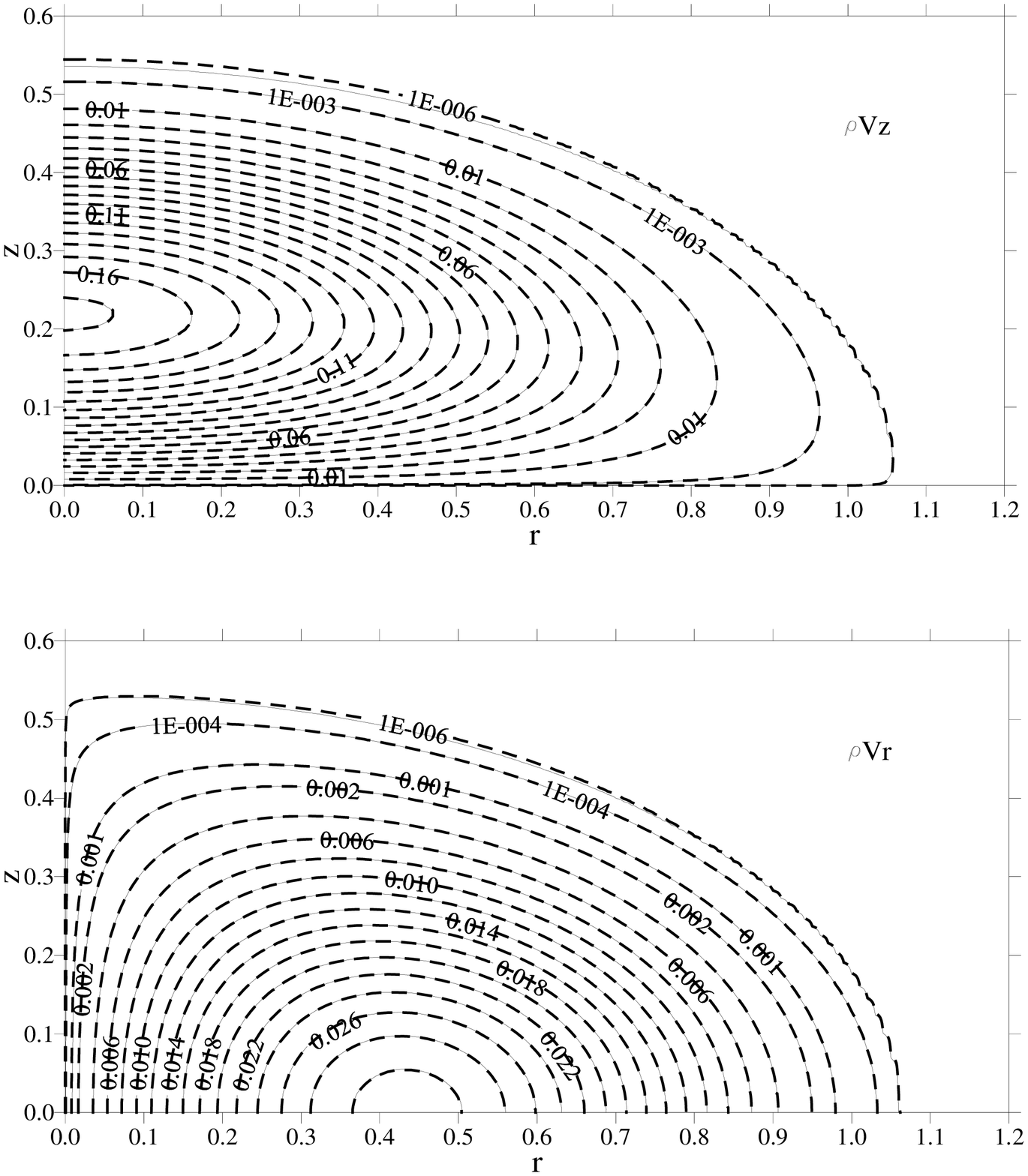}}
\caption{Anisimov's problem, momenta ($\protect\rho V_{z}$ \ and $\protect%
\rho V_{r}$) distribution. COS2 scheme versus the analytical solution. $%
\protect\sigma \equiv Z_{0}\diagup R_{0}=0.1$, time $t=0.4$, CFL number $%
C_{r}=1$, monotonicity parameter $\aleph =\varkappa =1$, spatial increments: 
$\Delta r=0.0025$, $\Delta z=0.00025$ if $0<t\leq 0.05$; $\Delta r=0.0025$, $%
\Delta z=0.0005$ if $0.05<t\leq 0.1$; $\Delta r=0.005$, $\Delta z=0.001$ if $%
0.1<t\leq 0.2$; $\Delta r=0.01$, $\Delta z=0.002$ if $0.2<t\leq 0.4$. Dashed
lines: numerical solution; Solid lines: analytical solution.}
\label{MZR1T0504}
\end{figure}
We observe (Figures \ref{DP1T0504}, \ref{MZR1T0504}) that the numerical and
analytical solutions are practically coincide, but in the vicinity of the
front, namely, for very small values of density.

\section{Appendix\label{Appendix1}}

\begin{proposition}
\label{Approximate derivative}Let us find the order of accuracy, $r$, in (%
\ref{C110}) if $d_{i}$ will be approximated by $\widetilde{d}_{i}$ with the
order of accuracy $s$, i.e. let 
\begin{equation}
d_{i}=\widetilde{d}_{i}+O\left( \left( \Delta x\right) ^{s}\right) .
\label{G90}
\end{equation}
Let $U\left( x\right) $ be sufficiently smooth, then we can write 
\begin{equation}
U_{i+1}=U_{i+05}+U_{i+05}^{\prime }\frac{\Delta x}{2}+\frac{1}{2}%
U_{i+05}^{\prime \prime }\left( \frac{\Delta x}{2}\right) ^{2}+O\left(
\left( \Delta x\right) ^{3}\right) ,  \label{HA300}
\end{equation}
\begin{equation}
U_{i}=U_{i+05}-U_{i+05}^{\prime }\frac{\Delta x}{2}+\frac{1}{2}%
U_{i+05}^{\prime \prime }\left( \frac{\Delta x}{2}\right) ^{2}+O\left(
\left( \Delta x\right) ^{3}\right) .  \label{HA310}
\end{equation}
Combining the equalities (\ref{HA300}) and \ref{HA310} we obtain 
\begin{equation}
U_{i+1}+U_{i}=2U_{i+05}+\left. \frac{\partial ^{2}U}{\partial x^{2}}\right|
_{i+05}\left( \frac{\Delta x}{2}\right) ^{2}+O\left( \left( \Delta x\right)
^{3}\right) .  \label{HA320}
\end{equation}
In a similar manner we write: 
\begin{equation}
d_{i+1}=U_{i+05}^{\prime }+U_{i+05}^{\prime \prime }\frac{\Delta x}{2}+\frac{%
1}{2}U_{i+05}^{\prime \prime \prime }\left( \frac{\Delta x}{2}\right)
^{2}+O\left( \left( \Delta x\right) ^{3}\right) ,  \label{HA330}
\end{equation}
\begin{equation}
d_{i}=U_{i+05}^{\prime }-U_{i+05}^{\prime \prime }\frac{\Delta x}{2}+\frac{1%
}{2}U_{i+05}^{\prime \prime \prime }\left( \frac{\Delta x}{2}\right)
^{2}+O\left( \left( \Delta x\right) ^{3}\right) .  \label{HA340}
\end{equation}
Subtracting the equations (\ref{HA330}) and (\ref{HA340}), we obtain 
\begin{equation}
\left. \frac{\partial ^{2}U}{\partial x^{2}}\right| _{i+05}=\frac{%
d_{i+1}-d_{i}}{\Delta x}+O\left( \left( \Delta x\right) ^{2}\right) .
\label{HA350}
\end{equation}
In view of (\ref{HA350}) and (\ref{G90}) we obtain from (\ref{HA320}) the
following interpolation formula 
\begin{equation}
U_{i+05}=\frac{1}{2}\left( U_{i+1}+U_{i}\right) -\frac{\Delta x}{8}\left( 
\widetilde{d}_{i+1}-\widetilde{d}_{i}\right) +O\left( \left( \Delta x\right)
^{4}+\left( \Delta x\right) ^{s+1}\right) .  \label{HA360}
\end{equation}
In view of (\ref{HA360}) we obtain that $r=\min \left( 4,s+1\right) .$
\end{proposition}

\begin{proposition}
\label{First to Second order scheme}Let us construct a second order scheme
based on operator-splitting techniques. We will, in fact, use the summarized
(summed) approximation method \cite[Section 9.3]{Samarskii 2001} to estimate
order of approximation. Consider the following equation 
\begin{equation}
\mathcal{P}\mathbf{u}\equiv \mathcal{P}_{1}\mathbf{u}+\mathcal{P}_{2}\mathbf{%
u}\equiv \frac{\partial \mathbf{u}}{\partial t}-L\mathbf{u}=0\mathbf{,\quad }%
\mathcal{P}_{k}\mathbf{u}\equiv \frac{1}{2}\frac{\partial \mathbf{u}}{%
\partial t}-L_{k}\mathbf{u},\ k=1,2,  \label{FS10}
\end{equation}
where $L_{k}$ is an operator, e.g. a differential operator, a real analytic
function, etc., acting on $\mathbf{u}\left( x,t\right) $. We approximate (%
\ref{FS10}) on the cell $\left[ x_{i-1},x_{i+1}\right] \times \left[
t_{n},t_{n+1}\right] $ by the following difference equation with the
accuracy $O(\left( \Delta x\right) ^{2}+\left( \Delta t\right) ^{2})$%
\begin{equation}
\Pi \mathbf{v}\equiv \frac{\mathbf{v}_{i}^{n+1}-\mathbf{v}_{i}^{n}}{\Delta t}%
-\Lambda _{1}\mathbf{v}^{n+0.5}-\Lambda _{2}\mathbf{v}^{n+0.5}=0,
\label{FS15}
\end{equation}
where it is assumed that the operator $L_{k}\mathbf{u}$ is approximated by
the operator $\Lambda _{k}\mathbf{u}$ with the accuracy $O(\left( \Delta
x\right) ^{2})$, i.e. 
\begin{equation}
\Lambda _{k}\mathbf{u}^{n+0.5}=\left( L_{k}\mathbf{u}\right)
_{i}^{n+0.5}+O\left( \left( \Delta x\right) ^{2}\right) .  \label{FS17}
\end{equation}
In view of the operator splitting idea, to the problem (\ref{FS15}) there
corresponds the following chain of difference schemes 
\begin{equation}
\Pi _{1}\mathbf{w}\equiv \frac{1}{2}\frac{\mathbf{w}_{i}^{n+0.5}-\mathbf{w}%
_{i}^{n}}{0.5\Delta t}-\Lambda _{1}\mathbf{w}_{1}^{n+0.5}=0,  \label{FS30}
\end{equation}
\begin{equation}
\Pi _{2}\mathbf{w}\equiv \frac{1}{2}\frac{\mathbf{w}_{i}^{n+1}-\mathbf{w}%
_{i}^{n+0.5}}{0.5\Delta t}-\Lambda _{2}\mathbf{w}_{2}^{n+0.5}=0.
\label{FS40}
\end{equation}
One can see from the above that the operator $\mathcal{P}_{k}\mathbf{u}$ is
approximated by $\Pi _{k}\mathbf{u}$ with the accuracy $O(\Delta t+\left(
\Delta x\right) ^{2})$%
\begin{equation}
\Pi _{1}\mathbf{u}_{i}^{n+0.5}=\left( \mathcal{P}_{1}\mathbf{u}\right)
_{i}^{n+0.5}-\frac{\Delta t}{8}\left( \frac{\partial ^{2}\mathbf{u}}{%
\partial t^{2}}\right) _{i}^{n+0.5}+O\left( \left( \Delta t\right)
^{2}+\left( \Delta x\right) ^{2}\right) ,  \label{FS60}
\end{equation}
\begin{equation}
\Pi _{2}\mathbf{u}_{i}^{n+0.5}=\left( \mathcal{P}_{2}\mathbf{u}\right)
_{i}^{n+0.5}+\frac{\Delta t}{8}\left( \frac{\partial ^{2}\mathbf{u}}{%
\partial t^{2}}\right) _{i}^{n+0.5}+O\left( \left( \Delta t\right)
^{2}+\left( \Delta x\right) ^{2}\right) .  \label{FS70}
\end{equation}
In view of (\ref{FS60})-(\ref{FS70}), the local truncation error \cite[p.
142]{LeVeque 2002}, $\psi $, on a sufficiently smooth solution $\mathbf{u}%
(x,t)$ to (\ref{FS10}) is found to be 
\begin{equation}
\psi =\Pi \mathbf{u}=\Pi _{1}\mathbf{u+}\Pi _{2}\mathbf{u}=  \label{FS75}
\end{equation}
\begin{equation}
\left( \mathcal{P}_{1}\mathbf{u+}\mathcal{P}_{2}\mathbf{u}\right)
_{i}^{n+0.5}+O\left( \left( \Delta t\right) ^{2}+\left( \Delta x\right)
^{2}\right) =O\left( \left( \Delta t\right) ^{2}+\left( \Delta x\right)
^{2}\right) .  \label{FS80}
\end{equation}
Thus, the implicit scheme, (\ref{FS30}), together with the explicit scheme, (%
\ref{FS40}), approximate (\ref{FS10}) with the second order.
\end{proposition}

\end{document}